\documentclass[twocolumn,secnumarabic,nobibnotes,superscriptaddress,aps,nofootinbib,pra,10pt]{revtex4-2}
\usepackage{color, xcolor, colortbl}
\usepackage{graphicx}
\usepackage{amsthm,amsmath,amssymb,amsfonts}
\usepackage{dcolumn}
\usepackage{url}
\usepackage{epstopdf}
\usepackage{enumitem}
\usepackage{algorithm}
\usepackage{algpseudocode}
\usepackage{bm}
\usepackage[caption=false]{subfig}
\usepackage{appendix}
\usepackage{multirow}
\usepackage{braket}
\usepackage[english]{babel}
\usepackage{hyperref}
\usepackage[capitalize]{cleveref}
\usepackage{xpatch}
\usepackage{tikz}
\usepackage{adjustbox}
\usepackage{xspace}

\newcommand{\Or}{\mathcal{O}}

\newcommand{\RR}{\mathbb{R}}

\renewcommand{\Re}{\operatorname{Re}}

\newcommand{\mc}[1]{\mathcal{#1}}

\newcommand{\norm}[1]{\left\lVert#1\right\rVert}

\newtheorem{thm}{\protect\theoremname}
\newtheorem{lem}[thm]{\protect\lemmaname}

\newtheorem{defn}[thm]{\protect\definitionname}

\newtheorem{result}[thm]{Result}

\providecommand{\definitionname}{Definition}
\providecommand{\assumptionname}{Assumption}
\providecommand{\corollaryname}{Corollary}
\providecommand{\lemmaname}{Lemma}
\providecommand{\propositionname}{Proposition}
\providecommand{\remarkname}{Remark}
\providecommand{\theoremname}{Theorem}
\usepackage{bbm}

\usepackage{silence}


\usetikzlibrary{fit}
\tikzset{%
  highlight/.style={rectangle,rounded corners,fill=blue!15,draw,fill opacity=0.3,thick,inner sep=0pt}
}

\usepackage{booktabs}
\newcolumntype{L}[1]{>{\raggedright\arraybackslash}p{#1}}
\usepackage{enumitem}

\renewcommand{\tablename}{Box}
\renewcommand{\thetable}{\arabic{table}}

\begin{document}

	\title{Ensemble-Based Quantum Signal Processing for Error Mitigation }

\author{Suying Liu}
\thanks{These authors contribute equally.}
    \affiliation{Joint Center for Quantum Information and Computer Science, University of Maryland, College Park, MD 20742, USA.} 
    \affiliation{Department of Computer Science, University of Maryland, College Park, MD 20742, USA.}

    \author{Yulong Dong}
    \thanks{These authors contribute equally.}
    \email[Electronic address: ]{dongyl@umich.edu}
\affiliation{Department of Electrical Engineering and Computer Science, University of Michigan, Ann Arbor, MI 48109, USA}

        \author{Dong An}
        \thanks{These authors contribute equally. This research was conducted during the author's postdoctoral fellowship at the University of Maryland.}
    \affiliation{Joint Center for Quantum Information and Computer Science, University of Maryland, College Park, MD 20742, USA.} 

\author{Murphy Yuezhen Niu}
	\email[Electronic address: ]{murphyniu@ucsb.edu}
	\affiliation{Google Quantum AI, Venice, California 90291, USA.}
  \affiliation{Department of Computer Science, University of California, Santa Barbara, California, 93106, USA.}

\begin{abstract}
Despite rapid advances in quantum hardware, noise remains a central obstacle to deploying quantum algorithms on near-term devices. In particular, random coherent errors that accumulate during circuit execution constitute a dominant and fundamentally challenging noise source. We introduce a noise-resilient framework for Quantum Signal Processing (QSP) that mitigates such coherent errors without increasing circuit depth or ancillary qubit requirements. Our approach uses ensembles of noisy QSP circuits combined with measurement-level averaging to suppress random phase errors in Z rotations. Building on this framework, we develop robust QSP algorithms for implementing polynomial functions of Hermitian matrices and for estimating observables, with applications to Hamiltonian simulation, quantum linear systems, and ground-state preparation. We analyze the trade-off between approximation error and hardware noise, which is essential for practical implementation under the stringent depth and coherence constraints of current quantum hardware. Our results establish a practical pathway for integrating error mitigation seamlessly into algorithmic design, advancing the development of robust quantum computing, and enabling the discovery of scientific applications with near- and mid-term quantum devices.

\end{abstract}

\maketitle

\section{Introduction}
Quantum computers promise exponential improvement over classical computers for solving computational problems ranging from quantum simulation to linear and nonlinear systems \cite{BerryChildsKothari2015,ChildsKothariSomma2017,LowChuang2017,GilyenSuLowEtAl2019,MotlaghWiebe2023}. Quantum Signal Processing (QSP)~\cite{LowChuang2017} and its generalizations, such as qubitization~\cite{LowChuang2019} and Quantum Singular Value Transformation (QSVT)~\cite{GilyenSuLowEtAl2019}, provide a unified framework~\cite{martyn2021grand} for implementing general linear transformations with sufficiently accurate polynomial approximations. By encoding operator dynamics into a collection of two-dimensional invariant subspaces and coherently transforming them through phase-controlled evolutions, QSP enables the efficient implementation of polynomial transformations using minimal ancillary resources and optimal algorithmic complexity \cite{LowChuang2017}. Due to its efficiency, generality and simplicity, QSP and its variants soon become a powerful quantum primitive and have been widely applied to designing quantum algorithms for Hamiltonian simulation problems~\cite{LowChuang2017,LowChuang2019,GilyenSuLowEtAl2019}, solving linear systems of equations~\cite{ChakrabortyGilyenJeffery2018,GilyenSuLowEtAl2019}, eigenstate filtering and preparation~\cite{LinTong2020,LinTong2020b}, and other tasks involving polynomial transformations of Hermitian matrices. 

An important limitation of existing QSP algorithms is their necessary assumption of negligible error rates compared to algorithmic accuracy. Without a fault-tolerant quantum computer and perfect gates, QSP-based algorithms can neither achieve optimality nor provide reliable results. Unfortunately, existing quantum computers face significant errors many magnitudes higher than QSP requirements due to control flaws and environmental coupling \cite{andersen2025thermalization}. These systems, known as Near-term Intermediate Scale Quantum (NISQ) systems, lack built-in error correction \cite{analogrobust2024}. While error cancellation has been proposed within the QSP framework \cite{tan2023error,tan2023perturbative}, these methods generally introduce significant quantum overhead and rely on assumptions that are unlikely to hold in practice. Beyond these approaches, error mitigation techniques are widely employed to enhance the performance of quantum algorithms on NISQ devices. Standard methods include zero-noise extrapolation~\cite{temme2017error}, which infers ideal observables by systematically varying noise strength, post-selection schemes that exploit system-specific symmetries~\cite{google2020hartree}, and protocols that leverage the intrinsic robustness of certain observables against specific noise sources~\cite{mi2021information, arute2020observation}. While effective in targeted settings, these techniques are typically tailored to particular observables or restricted noise models, which limits their applicability to general-purpose quantum algorithm implementations.

A major source of error in NISQ devices is phase error in Z rotations, notoriously difficult to eliminate, causing both coherent~\cite{lucero2010reduced,rol2020time} and incoherent errors~\cite{wudarski2023characterizing} that compromise quantum algorithm implementations. 
Such errors can distort simulations, affecting free fermionic terms and interactions between spins or charges~\cite{arute2020observation}, and limit the simulable system size without error correction~\cite{google2020hartree}. As NISQ technology advances toward tens of thousands of qubits, approaching the scale needed for QSP applications, developing QSP frameworks that accommodate these imperfections will be critical for enabling impactful applications. Achieving this requires replacing a separate, resource-intensive quantum error correction layer with efficient, in-algorithm error mitigation and correction schemes.

In this work, we tackle this critical bottleneck by proposing a new framework for QSP to mitigate and eventually eliminate the effect of systematic random errors. Specifically, we show that random phase misrotation error in Z gates can be corrected by leveraging ensembles of noisy QSP circuits through measurement-level averaging. Building on this idea, we develop an Ensemble-based QSP algorithms for polynomial transformations of Hermitian matrices and for observable estimation, with applications to Hamiltonian simulation, quantum linear systems, and ground-state preparation. We analyze the trade-off between approximation error and hardware noise. This favorable balance is essential for practical implementation under the stringent depth and coherence constraints of current quantum hardware. Consequently, our method not only bolsters the reliability and accuracy of QSP algorithms but also establishes a new paradigm for integrating error mitigation into quantum algorithm design. Through this integration, we pave the way for more robust and accurate quantum computing applications, where error mitigation is an inseparable part of the algorithmic design.

\section{Results}
\subsection{Quantum Signal Processing and Matrix Polynomials}
The goal of this work is to integrate quantum error correction directly into quantum signal processing (QSP) circuits, enabling the robust deployment of quantum algorithms under realistic noise models. To develop our main methods and state our principal theorems, we first introduce several essential components.

Many important quantum algorithms operate on non-unitary matrices, whereas quantum circuits can only implement unitary evolutions. Block-encoding~\cite{GilyenSuLowEtAl2019} provides a powerful framework to bridge this gap by embedding a general matrix $A$ into a larger unitary operator $U_A$. Specifically, a matrix $A$ is $(\alpha,m,\epsilon)$-block encoded in unitary $U_A$ if 
\begin{align}
    \| A - \alpha (\langle 0|^m \otimes I) U (|0\rangle^m \otimes I) \| \leq \epsilon,
\end{align}
where $\alpha$ is the normalization factor, $m$ is the number of ancilla qubits and $\epsilon$ is the block-encoding precision. 
With the help of the block-encoding model, we can perform basic linear algebra operations such as addition and multiplication for general matrices on a quantum computer~\cite{GilyenSuLowEtAl2019}.

Beyond these basic operations, QSP~\cite{LowChuang2019,GilyenSuLowEtAl2019,LowChuang2017} serves as a systematic method for implementing polynomial transformations $f(A)$ of a general matrix $A$ by interleaving the block encoding with simple rotation gates. Consider the reflection signal operator $R(x)$ \footnote{Here we adopt the reflection-based convention of QSP \cite[Corollary 8]{GilyenSuLowEtAl2019}, where the signal operator is defined as the reflection gate $R(x)$ in \eqref{eqn:R gate}, instead of the standard QSP formulation \cite[Theorem 3]{GilyenSuLowEtAl2019}, which utilizes a rotation-based signal operator.} which encodes the input $x$, the length-$d$ QSP sequence with a specific set of phase factors $\Phi = (\phi_1,\cdots,\phi_d)$, transforms the entry towards a complex polynomial $P(x)$, i.e.,
    \begin{equation}
         U(x) := \prod_{j = 1}^d e^{i \phi_j Z} R(x) = \left( \begin{array}{cc}
            P(x) & \cdot \\
            \cdot & \cdot
        \end{array} \right), 
    \end{equation}
    where 
    \begin{equation}
        R(x) = \left(\begin{array}{cc}
            x & \sqrt{1-x^2} \\
            \sqrt{1-x^2} & -x
        \end{array}\right).\label{eqn:R gate}
    \end{equation}
In particular, given a real polynomial $p(x)$ satisfying the conditions listed in \cref{lem:QSP}, there exists a set of phase factors such that quantum signal processing implements a corresponding complex polynomial $p(x)$ on a block of $U(x)$. Specifically, the resulting transformation satisfies $\bra{0} U(x) \ket{0} = P(x),$ where the real part of $p(x)$ coincides with the target polynomial~\cite{GilyenSuLowEtAl2019}. 

The QSP framework can be further extended to implement polynomial transformations of a linear operator. This extension, known as Quantum Singular Value Transformation (QSVT)~\cite{GilyenSuLowEtAl2019}, lifts the scalar QSP sequence to an operator-level transformation that effectively applies a polynomial $p(x)$ to each singular value of the linear operator within its corresponding invariant two-dimensional subspace. Specifically, given a Hermitian matrix $A$ with $\|A\| \le 1$ and its $(\alpha, m, \epsilon)$-block-encoding $U_A$, the block-encoding of the matrix function $P(A)$ is realized by sequentially applying alternating queries to $U_A$ and controlled-phase rotations, where the phase factors are computed as in \cref{lem:QSP}. Since the QSP sequence yields a complex polynomial $P(x)$, the block-encoding of the desired real polynomial $p(A) = \text{Re}(P(A))$ can be extracted by computing $\frac{1}{2}(P(A)+P(A)^{\dagger})$, implemented via Linear Combination of Unitaries (LCU) circuits. The quantum circuit for achieving the implementation is shown in Fig.\ref{fig:QSP}.

\subsection{An Ensemble Method against Random Coherent Errors}\label{sec:ensemble_QSP_method}
In experimental implementation, the quantum circuit may be subjected to quantum noises and errors. In quantum platforms, such as superconducting qubits, one major source of errors is the random misrotation in $Z$-rotation gates caused by qubit frequency noise~\cite{lucero2010reduced,rol2020time,wudarski2023characterizing}. Such random error brings the desired rotation angle to a random value that deviates from the ideal case and distorts the implementation results. Furthermore, due to its random nature, these errors are hard to calibrate, while its mitigation necessitates sophisticated quantum error correction modules with significant resource overheads. In this section, we propose a framework to mitigate such random coherent errors, drawing inspiration from the ensemble method in classical machine learning techniques. 

We first assume that the random coherent error occurs in the rotation gates used to implement QSP phase-factor rotations. A general case when block-encoding oracles are subjected to random errors is discussed later in this subsection. The random coherent error is described as 
\begin{equation}\label{RzrotationEq}
    R_z(\phi; \varepsilon) = e^{- i (\phi + \varepsilon) Z},
\end{equation}
where $\phi$ is the desired rotation angle and $\varepsilon$ is a random error with zero mean and variance $\nu$. A key observation is that the expectation value of the rotation gate coincides with the desired rotation up to a multiplicative factor, namely, $\mathbb{E}_\varepsilon(R_z(\phi; \varepsilon)) = \mathbb{E}_\varepsilon(\cos(\varepsilon) R_z(\phi; 0) - i \sin(\varepsilon) R_z(\phi; 0)) = \mathbb{E}_\varepsilon(\cos(\varepsilon)) R_z(\phi; 0)$. This motivates us to consider distilling the correct quantum information by averaging multiple experimental realizations. 
We show in~Algorithm~\ref{alg:block_encoding} how to accurately implement polynomial transformations of matrices with QSP circuits subjected to random coherent errors. Our approach is to first implement the QSP circuit repeatedly multiple times, and then linearly combine them using the quantum Linear Combination of Unitaries (LCU) technique~\cite{ChildsWiebe2012}. Consequently, the overall algorithm constructs a new block encoding which is an average of multiple noisy implementations. Here, we assume the dominant error in LCU is in the physical Z rotation and thus translate to the same additive angle errors in the phase factors.  Following the averaged argument of individual gates, such averaged noisy implementations of matrix polynomial should approximate the noiseless one. We quantify the algorithm performance as follows, where the formal version and complete proof can be found in \cref{thm:noisy_QSP_block_encoding}. 

\begin{result}[EnQSP for matrix polynomials]\label{res:QSP_block_encoding}
    Let $p \in \RR_d[x]$ be a bounded real polynomial. Under the assumptions in \cref{sec:error_model}, using 
    $M = \Or(e^{ d \nu} \epsilon^{-2})$ 
    noisy implementations, a block encoding of $\mathcal{P}$ can be derived so that 
    $\norm{\mc{P}/\alpha  - p(A)} \le \epsilon$ 
    with a constant probability of success. Here, 
    $\alpha = e^{\frac{1}{2}d \nu}$ 
    is a multiplicative scaling factor. The total number of queries to the matrix block encoding $U_A$ scales as 
    $\Or(dM) = \Or(d e^{d \nu} \epsilon^{-2})$.
\end{result}

Though the result is formulated as the case where only phase-angle rotations are subjected to random coherent errors, it can be generalized to the case when additional random coherent errors may occur in the block encoding $U_A$. In the general case with total $g$ operations subjected to random coherent errors, the scaling factor and the repetition become $\alpha = e^{g \nu}$ and $M = \Or(\alpha^2 \epsilon^{-2})$ respectively. 

We also remark that the inverse scaling factor $1 / \alpha$ represents the signal-to-noise ratio in the original QSP implementation. When decomposing all phase-angle rotations, we have $\mc{Q}(U_A, \vec{\phi}; \vec{\varepsilon}) = (\prod_j \cos(\varepsilon_j)) \mc{Q}(U_A, \vec{\phi}; \vec{0}) + \text{error}$. Here, the first term is the noiseless implementation, which gives $\mc{Q}(U_A, \vec{\phi}; \vec{0}) = U_{p(A)}$, and the second error term contains at least one sine component of the random error. Taking expectation over random coherent errors, we have $\mathbb{E}(\mc{Q}(U_A, \vec{\phi}; \vec{\varepsilon})) \approx \alpha^{-1} U_{p(A)}$. Consequently, when the signal-to-noise ratio is vanishingly small, i.e., the scaling factor $\alpha$ is extremely large, the correct quantum information in the noisy implementation is vanishingly small. Then, the corresponding error correction problem is ill-posed, where to the best of our knowledge, no error mitigation techniques could retrieve that vanishingly small correct information. Our result is also consistent because the resource overhead in \cref{res:QSP_block_encoding} becomes increasingly large. Our result is most efficient in the regime where the signal-to-noise ratio is constantly large $\alpha = \Or(1)$ or depends very weakly on system parameters, such as the number of qubits. 

It is also worth noting that our proposed method can be integrated into a larger quantum algorithm as a modular subroutine. Moving one step ahead, using $\Or(\alpha)$ number of queries to the block encoding of $\mc{P}$ as discussed in \cref{res:QSP_block_encoding}, a block encoding that is $\epsilon$-close to the desired polynomial transformation $p(A)$ can be constructed using oblivious amplitude amplification \cite{OAA}. This error-corrected implementation enables versatile downstream applications using the subroutine based on it.

\subsection{Balance between Algorithmic Errors and Implementation Errors}

As in \cref{res:QSP_block_encoding}, our method provides a way to mitigate random coherent errors to a certain level, which is referred to as \emph{implementation error}. In scientific applications, though in the ideal noiseless case, the algorithm is derived so that the computational results are close to the target solution. This error is referred to as \emph{algorithmic error}. In this subsection, we will show how to balance these error sources.

Many tasks in scientific applications are described by matrix functions. For example, preparing the ground state of a Hamiltonian is equivalent to implementing a mapping from the Hamiltonian matrix $H = \sum_i E_i \ket{\psi_i}\bra{\psi_i}$ to the projection onto its ground-state subspace $\ket{\psi_0}\bra{\psi_0}$. In terms of the transformation of eigenvalues, it means implementing a matrix function $f(H)$ where $f(x)$ is a filter function taking the value zero at all excited-state energies and taking the value one at the ground-state energy. Many Quantum Eigenvalue Transformation-based techniques, such as QSP, provide an elegant way to implement polynomial matrix transformations. It is well established in approximation theory that smooth functions can be approximated by polynomials. That means that we can find an approximation polynomial $p \in \RR_d[x]$ that is close to the target smooth function $\Vert p(x) - f(x) \Vert_\infty \le \epsilon_\text{alg}$. Here, the approximation error is referred to as the \emph{algorithmic error} which stands for the error raised from algorithm construction, though without any implementation error. The approximation efficiency in the case of smooth target functions is captured by the polynomial degree requirement $d = \Or(s \log(1 / \epsilon_\text{alg}))$ where $s$ represents some functions of the system parameters. For example, in the case of ground state preparation, $s = \norm{H} / \Delta$ where the spectral gap $\Delta$ is equal to the difference between the first excited state energy and ground state energy. Following an argument based on the triangle inequality, the total error is equal to the sum of algorithmic error and implementation error:
\begin{equation}
\begin{split}
    & \epsilon_\text{tot} := \Vert \mc{P}/\alpha - f(H) \Vert\\
    &\le \Vert \mc{P}/\alpha - p(H) \Vert + \Vert p(H) - f(H) \Vert \le \epsilon_\text{imp} + \epsilon_\text{alg}.
\end{split}
\end{equation}

When mitigating algorithmic error, a deeper circuit with more queries is favored, while the implementation error in the second term still bottlenecks the total error. However, deeper circuitry renders the scaling factor $\alpha$ larger, and therefore makes the error correction harder in the presence of a much smaller signal-to-noise ratio. Conversely, shorter circuits make error correction of implementation errors easier, but the algorithmic error dominates. Consequently, balancing these two error resources is important to make the error-corrected algorithm resource-efficient and error-resilient.

The number of repeated implementations scales as $M = \Or(e^{\Or(\nu s \log(1/\epsilon_\text{alg}))} \epsilon_\text{imp}^{-2}) = \Or(\epsilon_\text{alg}^{- \Or(\nu s)} \epsilon_\text{imp}^{-2}) = \Or(\epsilon_\text{tot}^{-(2 + \Or(\nu s))})$, and query depth scales as $d = \Or(s \log(1/\epsilon_\text{tot}))$ in the case of a balanced choice of errors. Though the dependency of system parameters occurs in the exponent of $M$-scaling, we note that there is also a preconstant of error parameter $\nu$. In the case that the error is sufficiently small, $\nu s \ll \Or(1)$, the actual scaling is reasonable. Since the system parameter is also closely related to query depth $s \sim d$, the requirement $\nu d \ll \Or(1)$ is consistent with the signal-to-noise requirement discussed in \cref{sec:ensemble_QSP_method} to avoid ill-posedness of the error mitigation problem.

\subsection{Applications in Scientific Computing}

\begin{table*}[htbp]
\begin{tabular}{c|c|c|c}
\hline
Task                                                                     & Query Depth & Ensemble Size (Classical Repetition) & Total Query Complexity \\ \hline
\begin{tabular}[c]{@{}c@{}}Hamiltonian \\ Simulation\end{tabular}        &    $\mathcal{O}\left( \|H\|T + \log(1/\epsilon) \right)$         &    $\mathcal{O}\left( e^{d\nu} \epsilon^{-2}\right)$                  &           $\mathcal{O}\left(\left( \|H\|T + \log(1/\epsilon) \right) e^{\mathcal{O}(d\nu)}\epsilon^{-2}\right)$                \\ \hline
\begin{tabular}[c]{@{}c@{}}Quantum Linear \\ System Problem\end{tabular} &       $\mathcal{O}\left(\kappa \log(\kappa/\epsilon)\right)$         &     $\mathcal{O}\left(e^{d\nu}\kappa^2 \epsilon^{-2}  \right)$                                   &        $\mathcal{O}\left(\log(\kappa/\epsilon)e^{\mathcal{O}(d\nu)}\kappa^3 \epsilon^{-2}  \right)$                  \\ \hline
\begin{tabular}[c]{@{}c@{}}Ground State \\ Preparation\end{tabular}      &       $\mathcal{O}\left(\Delta^{-1} \log(1/\epsilon)\right)$          &   $\mathcal{O}\left(e^{3d\nu/2}\gamma^{-1} \epsilon^{-2}  \right)$                                     &  $\mathcal{O}\left(\Delta^{-1} \log(1/\epsilon)e^{\mathcal{O}(d\nu)}\gamma^{-1} \epsilon^{-2}  \right)$                       \\ \hline
\end{tabular}
\caption{Complexity results for scientific computing applications. }
\label{tab:my-table}
\end{table*}

Our algorithms are broadly applicable across scientific computing tasks due to the generality of the QSP framework. To demonstrate its performance in practice under system noise, here we study three key applications: Hamiltonian simulation, quantum linear systems problem, and ground state preparation problem.

\textit{1. Hamiltonian simulation}:  
Given a Hamiltonian $H$, the Hamiltonian simulation task aims to prepare a quantum state $e^{-iHT}\ket{\psi}$ at a given time $T$. Since the unitary evolution $e^{-iHT}$ admits an accurate low-degree polynomial approximation, quantum signal processing (QSP) and quantum singular value transformation (QSVT) have been employed to construct optimal Hamiltonian simulation algorithms~\cite{LowChuang2017,GilyenSuLowEtAl2019}. In these frameworks, implementing $e^{-iHT}$ to algorithmic accuracy $\epsilon_{\text{alg}}$ requires truncating the polynomial expansion to degree $d = \mathcal{O}(\|H\|T + \log(1/\epsilon_{\text{alg}}))$. In the presence of noise, we employ the EnQSP algorithm to mitigate implementation errors. To achieve an implementation accuracy of $\epsilon_{\mathrm{imp}}$, $M=\mathcal{O}(e^{d\nu} \epsilon_{\mathrm{imp}}^{-2})$ repetitions are required. As discussed in the last section, a smaller algorithmic error $\epsilon_{\text{alg}}$ corresponds to a deeper circuit while adding a larger repetition overhead to mitigate implementation error $\epsilon_{\mathrm{imp}}$. With a balanced choice of both errors, i.e., $\epsilon:= \epsilon_{\text{alg}} =\epsilon_{\mathrm{imp}}$, the complexity of our algorithms is summarized below.

\begin{result}[Informal version of~\cref{thm:app_Hsim}]
    With only noisy access to Z rotations, we can approximate the state $e^{-iHT}\ket{\psi}$ with ensemble size $\mathcal{O}(e^{d\nu} \epsilon^{-2})$ and quantum circuits with depth $d = \mathcal{O}\left( \|H\|T + \log(1/\epsilon) \right)$, the total query complexity is $\mathcal{O}\left(\left( \|H\|T + \log(1/\epsilon) \right) e^{\mathcal{O}(d\nu)}\epsilon^{-2}\right)$. 
\end{result}
The query complexity is $\mathcal{O}\left(\left( \|H\|T + \log(1/\epsilon) \right) \epsilon^{-2}\right)$, polynomial scaling in system parameters (evolution time $T$ and the Hamiltonian spectral norm $\|H\|$) and has polynomial precision dependence when signal-to-noise ratio $1/\alpha = e^{-d\nu}$ is constantly large.

\textit{2. Quantum linear systems problem (QLSP)}: 
The goal of QLSP is to prepare a quantum state encoding $x/\|x\|$, where $x = A^{-1}\ket{b}$ is the possibly unnormalized solution of the linear system with an inevitable matrix $A$ and a quantum state $\ket{b}$. 
Without loss of generality, we further assume that $A$ is a Hermitian matrix with condition number $\kappa$, because any linear system of equations with a general invertible matrix can be reduced to an equivalent one with a Hermitian matrix by the dilation trick~\cite{HarrowHassidimLloyd2009}. The function $A^{-1}$ can be accurately implemented via polynomial approximation using the QSP construction introduced in~\cite{GilyenSuLowEtAl2019}, where truncating the polynomial approximation to degree $d = \mathcal{O}(\kappa \log(\kappa / \epsilon_{\mathrm{alg}}))$ ensures an algorithmic accuracy of $\epsilon_{\mathrm{alg}}$. Further, considering the coherent random errors in the QSP circuits, we can get a robust implementation by utilizing the EnQSP algorithm. With $M =\mathcal{O}(e^{d\nu} \epsilon_{\mathrm{imp}}^{-2}) $ repetition overheads, we can ensure $\epsilon_{\mathrm{imp}}$ implementation precision. Balancing the trade-off between the algorithm error and the implementation error,i.e., setting $\epsilon = \epsilon_{\mathrm{alg}}=\epsilon_{\mathrm{imp}}$, we show the complexity of correcting the noise in solving the QLSP as follows.  
\begin{result}[Informal version of~\cref{thm:app_QLSP}]
    With only noisy access to Z rotations, we may approximate the state $A^{-1}\ket{\psi}/\|A^{-1}\ket{\psi}\|$ with ensemble size $\mathcal{O}\left(e^{d\nu}\kappa^{2}\epsilon^{-2}\right)$ and quantum circuits with depth $d = \mathcal{O}\left(\kappa \log(\kappa/\epsilon)\right)$, 
    The total query complexity is $\mathcal{O}\left(\log(\kappa/\epsilon)e^{\mathcal{O}(d\nu)}\kappa^3 \epsilon^{-2}  \right)$ in both cases. 
\end{result}
Again, our algorithms can achieve favorable scalings. In the regime where the signal-to-noise ratio is constantly large, i.e., $1/\alpha = \Or(1)$, the overall complexity becomes $\mathcal{O}\left(\log(\kappa/\epsilon)\kappa^3 \epsilon^{-2}  \right)$ which scales polynomial with respect to both the system parameter (i.e., the condition number $\kappa = \|A\|\|A^{-1}\|$) and the total precision $\epsilon$. Although the overall complexity dramatically grows as the condition number of $A$ increases, the circuit depth for preparing the solution state remains the same as the noiseless QSP circuit even for ill-conditioned systems with large $\kappa$. \\

\textit{3. Ground state preparation}:
For the ground state preparation problem, the objective is to prepare the ground state $\ket{\psi_0}$ of a properly shifted and scaled Hamiltonian $H$, starting from an initial state $\ket{\phi}$ that has an overlap $\gamma = \braket{\phi | \psi_0}$ with the ground state. This problem can be addressed by implementing a filter function $f(x)$ that preserves the ground-state component while suppressing contributions from all excited energies, following the design in~\cite{Dong_Lin_Tong_2022}. The particular filter function, which incorporates a cosine transformation in its argument, can be approximated to accuracy $\epsilon_{\mathrm{alg}}$ by a polynomial of degree $d = \mathcal{O}\!\left( \Delta^{-1} \log(1/\epsilon_{\mathrm{alg}}) \right)$, where $\Delta$ is the spectral gap of the shifted and scaled Hamiltonian $H$. 

This consine transformation can be realized by quantum eigenvalue transformation of unitary matrices (QETU)~\cite{Dong_Lin_Tong_2022}. Let $H$ be a properly shifted and scaled Hamiltonian so that it is positive semidefinite and bounded $0 \prec H \prec \pi I$. In terms of Hamiltonian evolution, the scaling and shifting can be implemented by changing the evolution time and adding phase gates, respectively. Suppose the controlled Hamiltonian evolution is accessible, namely, $\ket{0}\bra{0} \otimes e^{i H} + \ket{1}\bra{1} \otimes e^{- i H}$. Then, sandwiching the controlled time evolution with two Hadamard gates on the control qubit, we can access a block encoding of the cosine-transformed Hamiltonian $\cos(H)$, which we refer to as a \textit{cosine block encoding}. By employing the cosine block-encoding realized through QSP circuits with depth $d$, we can efficiently implement the corresponding filter function for the ground state preparation problem. Then we further consider the noisy QSP circuits and utilize the EnQSP algorithm to mitigate the implementation error to achieve accuracy $\epsilon_{\mathrm{imp}}$ by $M = \mathcal{O}(e^{d\nu}\epsilon_{\mathrm{imp}}^{-2})$ repetitions. Under the balanced error condition ($\epsilon:= \epsilon_{\mathrm{alg}} = \epsilon_{\mathrm{imp}}$), the complexity of the ground state preparation problem with noise can be summarized as follows.

\begin{result}[Informal version of~\cref{thm:app_GSP}]
    With only noisy access to Z rotations, we may approximate the ground state $\ket{\psi_0}$ with ensemble size $\mathcal{O}\left(e^{3d\nu/2}\gamma^{-1}\epsilon^{-2}\right)$ and quantum circuits with depth $d = \mathcal{O}\left( \Delta^{-1}\log(1/\epsilon) \right)$. 
    The total query complexity is $\mathcal{O}\left( \Delta^{-1}\log(1/\epsilon)e^{\mathcal{O}(d\nu)}\gamma^{-1} \epsilon^{-2}  \right)$ for preparing the ground state. 
\end{result}

We also achieve similarly favorable scaling at this task. When the signal-to-noise ratio $1/\alpha = \Or(1)$, the total query complexity becomes $\mathcal{O}\left( \Delta^{-1}\log(1/\epsilon)\gamma^{-1} \epsilon^{-2}  \right)$, polynomial depending on system parameters (i.e., the spectral gap $\Delta$ of $H$ and the initial overlap $\gamma$) and precision $\epsilon$.  

\subsection{Generalizations of the Ensemble QSP}\label{sec:intro_additional_results}

In many applications, only some scalar observables are of interest instead of preparing a full quantum state. As the degree of freedom of observables is much smaller than that of quantum states in the Hilbert space, complex control operations in the LCU implementation are not necessary in a simpler case. In this subsection, we present an alternative classical ensemble average to replace the quantum average based on LCU and reduce the implementation cost. 

First consider the observable $\braket{\psi|p(A)^{\dagger}O p(A)|\psi}$ for a Hermitian matrix $O$ and a known quantum state $\ket{\psi}$. Here, $p(A)$ stands for the state preparation using QSP. For example, when estimating the density of states in the low-energy subspace, $p(A)$ is the projection onto the low-energy subspace. 

The algorithm is still based on using the average to approximate the desired quantity. 
However, we compute the average at the end on classical devices instead of performing the average coherently. 
The algorithm independently runs for multiple times the short-depth circuit in~\cref{fig:random_Hadamard}, which has exactly the same structure as the Hadamard test for non-unitary matrices proposed in~\cite{TongAnWiebe2021}, but the operator becomes inconsistent among different runs due to the noise. 
Then we classically average all the measurement outcomes, which is provably an accurate estimate of the desired observable. 
The cost of our algorithm is summarized in the following result. 

\begin{result}[Informal version of~\cref{thm:noisy_QSP_observable}]\label{res:QSP_observable}
    With only noisy access to Z rotations, we may estimate the observable $\braket{\psi|p(A)^{\dagger} O p(A)|\psi}$ for $d$-degree polynomial $p$ and Hermitian $O$ with at most $\epsilon$ error and constant probability. 
    The algorithm requires $\mathcal{O}( c^{-4d} \epsilon^{-2})$ independent runs of a circuit with depth $\mathcal{O}(d)$, where $c = \mathbb{E} \cos e < 1$ is the expectation of the cosine of the additive noise in the phase factors. 
\end{result}

Similar to~\cref{res:QSP_block_encoding}, our algorithm for estimating observables can also potentially achieve exponential speedup over classical counterparts in terms of dimension and scale quadratically in the inverse precision. 
The overall complexity is still exponential in the degree $d$. 
However, now such a computational overhead only appears to be a large number of repeats of a short-depth circuit. 
This circuit only has $\mathcal{O}(d)$ circuit depth, which is exactly the same as the noiseless QSP circuit. 
Therefore,~\cref{res:QSP_observable} shows that, if we are interested in obtaining the observables, then the noises in Z rotations can be corrected without any extra quantum resources, despite many more repeats of the same circuit. This algorithmic approach is also broadly applicable to a range of scientific computing tasks, with particular emphasis on accurately extracting observables in the presence of coherent system errors. We analyze the algorithmic scaling in detail for the same three representative applications: Hamiltonian simulation, the quantum linear system problem, and the ground state preparation problem in~\cref{thm:app_Hsim},\cref{thm:app_QLSP}, and \cref{thm:app_GSP}, respectively.

We remark that our algorithm for estimating observables from noisy QSP circuits can indeed work in the general scenario where the observable is random with Hermitian expectation. 
We believe this subroutine might be of independent interest and will discuss this separately below in~\cref{sec:intro_additional_results}. 

In designing our algorithms for estimating the observables of a noisy QSP circuit, we use a generalized version of the quantum Hadamard test. 
This generalized Hadamard test works for estimating any random observables beyond those obtained by QSP, so it might be of independent interest and will be separately discussed below. 

Our goal is to estimate the observable $\braket{\psi|\mathbb{E}\widetilde{O}|\psi}$ for a known quantum state $\ket{\psi}$ and a matrix $\widetilde{O}$. 
If $\widetilde{O}$ is deterministic and unitary, then we can use the Hadamard test together with the amplitude estimation to achieve Heisenberg limit. 
A generalization has been proposed in~\cite{TongAnWiebe2021}, enabling us to estimate the observable when $\widetilde{O}$ is deterministic but non-unitary. 
Both algorithms need to implement a Hadamard test circuit involving a single application of controlled $\widetilde{O}$ for many times. 

In our case, we consider the circuit of $\widetilde{O}$ with randomness. Such randomness may come from hardware errors, or the ideal observable operator is intrinsically random. 
In this case, we show that the standard Hadamard test still works -- we may run the Hadamard test circuit with random implementation of $\widetilde{O}$ as if it is deterministic, and estimate the observable using the outcomes of multiple runs in a standard way. 
The reason why the standard Hadamard test still works is that, although the probability distribution of the outcome conditioned on a specific sample of $\widetilde{O}$ varies among different runs of the Hadamard test circuit, the total probability distribution of the outcomes remains the same by the law of total probability. 
We describe the entire approach in~Algorithm~\ref{alg:random_Hadamard}, and show its validity and efficiency below. 

\begin{result}[Informal version of~\cref{thm:random_Hadamard}]
    Let $\widetilde{O}$ be a random matrix with Hermitian expectation. 
    Then we can estimate $\braket{\psi|\mathbb{E}\widetilde{O}|\psi}$ with error at most $\epsilon$ by running the Hadamard test circuit for $\mathcal{O}\left(\epsilon^{-2}\right)$ times. 
\end{result}

\section{Discussion}
In this work, we propose a new framework for quantum signal processing (QSP) that mitigates random coherent errors, enabling robust implementations of QSP circuits on noisy intermediate-scale quantum (NISQ) devices. By leveraging ensembles of randomized QSP circuits, we efficiently correct phase errors in Z rotations, which are typically recognized as the dominant noise source in NISQ systems and have been shown to be uncorrectable without using additional quantum resources~\cite{tan2023error}.

More specifically, we develop noise-resilient QSP algorithms for implementing polynomial functions of Hermitian matrices and for estimating observables, and apply these algorithms to a range of scientific computing tasks, including Hamiltonian simulation, the quantum linear systems problem, and ground-state preparation. We further analyze the trade-off between algorithmic error and implementation error, balancing between the algorithmic approximation accuracy and the accumulated hardware noise. Notably, the overhead required to mitigate systematic random misrotation errors scales favorably with the system size when the signal-to-noise ratio is constantly large. The algorithm requires only a polynomial number of repetitions of equal-depth QSP circuits compared to the noiseless case. This feature is particularly important given the stringent constraints on circuit depth and coherence time in current quantum hardware.

At the same time, our approach has several limitations that point to natural directions for future work. Our proposed algorithm suppresses the effects of $Z$-gate misrotations by averaging over independent and identically distributed (i.i.d.) noise realizations through classical post-processing. A possible alternative is to coherently manipulate multiple parallel noisy QSP circuits simultaneously~\cite{martyn2025parallelQSP}, enabling a space–time trade-off that can potentially reduce the repetition overhead of the algorithm. Moreover, the noise model considered in this work is restricted to additive phase noise in the Z-rotation gates. While this type of noise is indeed one of the dominant sources of quantum algorithm errors in current NISQ devices, extending our framework to address other leading sources of errors, such as crosstalk and amplitude damping errors~\cite{ash2020experimentalcrosstalk,google2020hartree}, is an important direction for future investigation. It is also of practical relevance to test our algorithms on real devices and examine whether we may already obtain any useful quantum advantage with available NISQ hardware, which would further clarify the regimes in which our noise-resilient QSP techniques provide tangible benefits.

\section*{Acknowledgements}

DA acknowledges support from the Department of Defense through the Hartree Postdoctoral Fellowship at QuICS, and the National Science Foundation through QLCI grants OMA-2120757. M. N. is
supported by the U.S. National Science Foundation grant CCF-2441912(CAREER), Air Force Office of Scientific Research under award number FA9550-25-1-0146, and   the U.S. Department of Energy,  Office of  Advanced Scientific Computing Research  under Award Number DE-SC0025430.

\bibliography{main}

\clearpage
\newpage
\appendix
\onecolumngrid

\begin{center}
    {\Large \bf Appendix}
\end{center}

The appendix of this paper is organized as follows. 
We first briefly review the concepts and propositions of block-encoding and QSP in~\cref{sec:prelim}. 
In~\cref{sec:noisy_QSP}, we present the formal definition of our noisy model and studies the expected outcome of the noisy QSP circuit. 
Based on that, we design and analyze the complexity of noise-resilient QSP algorithms for implementing polynomials of Hermitian matrices and estimating observables in~\cref{sec:algorithm}, and apply our algorithms to Hamiltonian simulation, linear system problems, and ground state preparation problems in~\cref{sec:applications}. 

\section{Preliminaries}\label{sec:prelim}

In this section, we briefly summarize preliminary results on block-encoding of a matrix and QSP. 
We refer to~\cite{GilyenSuLowEtAl2019} for comprehensive discussions. 

\subsection{Block-encoding}

Block-encoding is a powerful model for encoding a general matrix as an enlarged unitary operator. 
Intuitively, for a matrix $A$ with $\|A\| \leq 1$, its block-encoding is a unitary matrix $U_A$ in a higher-dimensional Hilbert space such that 
\begin{equation}
    U_A = \left( \begin{array}{cc}
        A & \cdot \\
        \cdot & \cdot
    \end{array} \right). 
\end{equation}
Its formal definition is given below. 
\begin{defn}[Block-encoding]
    Let $A$ be a matrix in $2^n\times 2^n$ dimension. Then a $2^{n+a}\times 2^{n+a}$ dimensional unitary matrix $U_A$ is called the block-encoding of $A$, if 
 $A = \bra{0}^{\otimes a} U_A \ket{0}^{\otimes a}$.  
\end{defn}

With the block-encoding model, we may perform basic linear algebra operations for general matrices on quantum computers. 
The following results from~\cite{GilyenSuLowEtAl2019} show how we may add and multiply several matrices. 
The matrix addition technique is indeed a generalized version of the quantum LCU subroutine. 

\begin{lem}
    Let $A = \sum_{j=0}^{m-1} y_j A_j$ and $\|y\|_1 \leq \beta$. 
    Suppose that $(P_L,P_R)$ is a pair of unitaries such that $P_L \ket{0} = \sum_{j=0}^{m-1} c_j\ket{j}$ and $P_R \ket{0} = \sum_{j=0}^{m-1} d_j\ket{j}$ with $y_j = \beta c_j^* d_j$ for all $j$, and $W = \sum_{j=0}^{m-1} \ket{j}\bra{j} \otimes U_j$ where $U_j$ is a block-encoding of $A_j$. 
    Then we can construct a block-encoding of $\frac{1}{\beta}A$ using each of $W$, $P_L^{\dagger}$ and $P_R$ once. 
\end{lem}

\begin{lem}
    If $U_A$ and $U_B$ are block-encodings of $A$ and $B$ with $n_A$ and $n_B$ many ancilla qubits, respectively, then $(I_{n_B}\otimes U_A)(I_{n_A}\otimes U_B)$ is a block-encoding of $AB$. 
\end{lem}

\subsection{Quantum signal processing and qubitization}

Let $p(x)$ be a real polynomial such that $|p(x)| \leq 1$ for $x \in [-1,1]$ and $\deg(p) = d$. 
QSP gives a unitary representation of this polynomial.  
Here we use the reflection version which is more suitable for block-encoding. 
\begin{lem}[{\cite[Corollary 10]{GilyenSuLowEtAl2019}}]\label{lem:QSP}
    Let $p(x)$ be a given real polynomial such that 
    \begin{enumerate}
        \item $|p(x)| \leq 1$ for all $x \in [-1,1]$, 
        \item $\deg(p) = d$, and 
        \item $p(x)$ has parity $d\mod 2$. 
    \end{enumerate}
    Then, there exists a complex polynomial $P(x)$ with $\Re(P(x)) = p(x)$ and a set of phase factors $\Phi = (\phi_1,\cdots,\phi_d)$, such that 
    \begin{equation}\label{eqn:QSP_sequence}
         e^{i \phi_d Z} R(x) e^{i \phi_{d-1} Z} R(x) \cdots e^{i \phi_1 Z} R(x) = \left( \begin{array}{cc}
            P(x) & \cdot \\
            \cdot & \cdot
        \end{array} \right), 
    \end{equation}
    where 
    \begin{equation}
        R(x) = \left(\begin{array}{cc}
            x & \sqrt{1-x^2} \\
            \sqrt{1-x^2} & -x
        \end{array}\right). 
    \end{equation}
\end{lem}

\begin{figure}
\includegraphics[width=\linewidth]{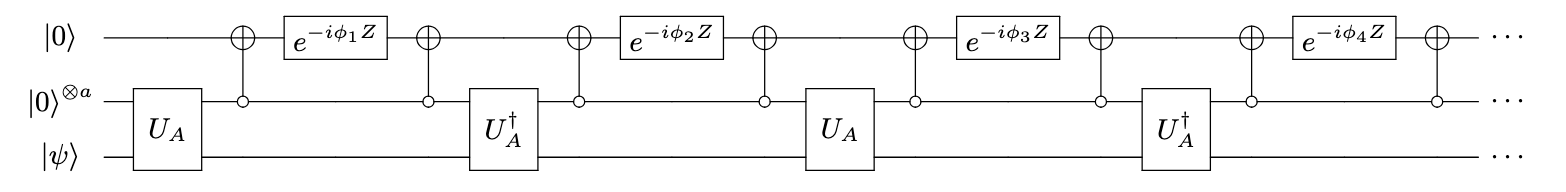}
    \caption{ Quantum circuit for block-encoding $P(A)$ using QSP and qubitization. }
    \label{fig:QSP}
\end{figure}

Let $A$ be a Hermitian matrix with $\|A\|\leq 1$, and we are given its block-encoding denoted by $U_A$. 
We are interested in implementing the matrix function $p(A)$. 
Notice that it suffices to focus on block-encoding $P(A)$ where $P(x)$ is the corresponding complex polynomial in~\cref{lem:QSP}, then a block-encoding of $p(A)$ can be obtained by using LCU to compute $\frac{1}{2}(P(A)+P(A)^{\dagger})$. 
For the block-encoding of $P(A)$, it can be constructed using the qubitization technique together with the QSP representation, by alternatively applying $U_A$ (or $U_A^{\dagger}$) and controlled rotations. 
The circuit is given in~\cref{fig:QSP}, where $\phi_j$'s are the phase factors in~\cref{lem:QSP}, and we formally state its validity as follows.

\begin{lem}[{\cite[Theorem 17]{GilyenSuLowEtAl2019}}] \label{lem:qubitization}
    Let $A$ be a Hermitian matrix with $\|A\|\leq 1$, and $U_A$ be its block-encoding. 
    Let $P(x)$ be a $d$-degree complex polynomial and $\Phi=(\phi_1,\cdots,\phi_d)$ be the corresponding phase factors, as given in~\cref{lem:QSP}. 
    Then, the quantum circuit in~\cref{fig:QSP} gives a block-encoding of $P(A)$. 
\end{lem}

\section{Noisy QSP}\label{sec:noisy_QSP}

\subsection{Error model}\label{sec:error_model}
 
Now let us suppose that the phase factors are shifted due to the imperfect implementation of the $Z$ gates. 
Specifically, suppose that the real implementation of QSP is the circuit in~\cref{fig:QSP} but with a different set of phase factors $\widetilde{\Phi} = (\widetilde{\phi}_1,\cdots,\widetilde{\phi}_d)$, yielding the corresponding complex polynomial $\widetilde{P}(x)$. 
We would like to explore the relation between the ideal $P(x)$ and the actual $\widetilde{P}(x)$, and investigate potential ways of correcting the errors. 

For the errors, let us assume that 
\begin{equation}
    \widetilde{\phi}_j = \phi_j + e_j. 
\end{equation}
Here $e_j$'s are independent random variables such that the probability distribution function of $e_j$ is an even function. 
Examples include $0$-mean Gaussian distribution. 
Notice that we do not need all $e_j$'s to be identically distributed. This error model is broad enough to describe both coherent gate  error~\cite{lucero2010reduced} and the dominant Gaussian low-frequency error in superconducting qubit frequency~\cite{wudarski2023characterizing}. With the proper use pulse sequence, the bias and distribution of such Z phase error can be tailored to an unbiased Gaussian distribution with a properly designed pulse sequence~\cite{degen2017quantum,choi2020robust}. For example, bias of the phase error can be reduced to arbitrarily to zero using CPMG-like sequence\cite{meiboom1958modified}; and the correlated low frequency noise can be randomized into white-noise through sample re-ordering or more sophisticated quantum control optimization~\cite{soare2014experimental}. Consequently, dominant Z phase error can be described by a zero bias Gaussian distributed additive error for realistic quantum gate engineering in superconducting qubits~\cite{niu2019universal}.

\subsection{Expectation of the noisy circuit}

A key observation of the noisy QSP is that the expectation of the noisy circuit still contains the information of the ideal $P(A)$ with an extra rescaling factor. 
We formally state and prove it in the next theorem. 

\begin{thm}\label{thm:QSP_noise_phase_shift}
    Let $A$ be a Hermitian matrix with $\|A\| \leq 1$, and we are given its block-encoding denoted by $U_A$. 
    Let $P(x)$ be a $d$-degree complex polynomial and $\Phi = (\phi_1,\cdots,\phi_d)$ be its QSP phase factors. 
    Suppose that the actual implementation of QSP is a circuit $\widetilde{U}$ with noisy $Z$-rotation, in the sense that the ideal rotation $e^{-i \phi_j Z }$ is actually implemented as $e^{-i \widetilde{\phi}_j Z }$ where 
    \begin{equation}
        \widetilde{\phi}_j = \phi_j + e_j, 
    \end{equation}
    and $e_j$'s are independent random variables with even probability distribution functions. 
    Then, 
    \begin{equation}
        \mathbb{E} \widetilde{U} = \left( \begin{array}{cc}
            \left(\prod_{j=1}^d \mathbb{E} \cos e_j\right) P(A) & \cdot \\
            \cdot & \cdot
        \end{array} \right). 
    \end{equation}
\end{thm}

\begin{proof}
    Let $\widetilde{P}(x)$ be a complex polynomial associated with the random phase factors $\widetilde{\Phi} = (\widetilde{\phi}_1,\cdots,\widetilde{\phi}_d)$. 
    We first study the relation between $\widetilde{P}(x)$ and $P(x)$. 
    According to~\cref{lem:QSP}, we have 
    \begin{equation}\label{eqn:proof_expectation_eq1}
        e^{i \phi_d Z} R(x) e^{i \phi_{d-1} Z} R(x) \cdots e^{i \phi_1 Z} R(x) = \left( \begin{array}{cc}
            P(x) & \cdot \\
            \cdot & \cdot
        \end{array} \right)
    \end{equation}
    and 
    \begin{equation}\label{eqn:proof_expectation_eq2}
        e^{i \widetilde{\phi}_d Z} R(x) e^{i \widetilde{\phi}_{d-1} Z} R(x) \cdots e^{i \widetilde{\phi}_1 Z} R(x) = \left( \begin{array}{cc}
            \widetilde{P}(x) & \cdot \\
            \cdot & \cdot
        \end{array} \right). 
    \end{equation}
    Take the expectation of~\cref{eqn:proof_expectation_eq2} and notice that different random phase factors $\widetilde{\phi}_j$'s are independent, then we have 
    \begin{equation}\label{eqn:proof_expectation_eq3}
        \left( \begin{array}{cc}
            \mathbb{E}\widetilde{P}(x) & \cdot \\
            \cdot & \cdot
        \end{array} \right) = \mathbb{E} \left(e^{i \widetilde{\phi}_d Z}\right) R(x) \mathbb{E} \left( e^{i \widetilde{\phi}_{d-1} Z} \right) R(x) \cdots \mathbb{E} \left(e^{i \widetilde{\phi}_1 Z}\right) R(x). 
    \end{equation}
    For each $\mathbb{E} \left(e^{i \widetilde{\phi}_j Z}\right)$, we have 
    \begin{align}
        \mathbb{E} \left(e^{i \widetilde{\phi}_j Z}\right) & = \left( \begin{array}{cc}
            \mathbb{E} e^{i \widetilde{\phi}_j } & 0 \\
            0 & \mathbb{E} e^{ - i \widetilde{\phi}_j }
        \end{array} \right)  \\
        & = \left( \begin{array}{cc}
            e^{i \phi_j } \mathbb{E} e^{i e_j } & 0 \\
            0 & e^{- i \phi_j } \mathbb{E} e^{ - i e_j }
        \end{array} \right) \\
        & = \left( \begin{array}{cc}
            e^{i \phi_j } \left( \mathbb{E} \cos e_j + i \mathbb{E} \sin e_j \right) & 0 \\
            0 & e^{- i \phi_j } \left( \mathbb{E} \cos e_j - i \mathbb{E} \sin e_j \right)
        \end{array} \right) \\
        & = \left( \begin{array}{cc}
            e^{i \phi_j } \left( \mathbb{E} \cos e_j \right) & 0 \\
            0 & e^{- i \phi_j } \left( \mathbb{E} \cos e_j \right)
        \end{array} \right) \\
        & =  \left( \mathbb{E} \cos e_j \right) e^{i \phi_j Z}, 
    \end{align}
    where the second to the last equation is because $e_j$ has even distribution function so $\mathbb{E} f(e_j) = 0$ for any odd function $f$. 
    Plugging this equation back to~\cref{eqn:proof_expectation_eq3} gives 
    \begin{equation}
         \left( \begin{array}{cc}
            \mathbb{E}\widetilde{P}(x) & \cdot \\
            \cdot & \cdot
        \end{array} \right) =  \left( \prod_{j=1}^d \mathbb{E} \cos e_j \right) e^{i \phi_d Z} R(x) e^{i \phi_{d-1} Z} R(x) \cdots e^{i \phi_1 Z} R(x) = \left( \prod_{j=1}^d \mathbb{E} \cos e_j \right) \left( \begin{array}{cc}
            P(x) & \cdot \\
            \cdot & \cdot
        \end{array} \right), 
    \end{equation}
    which indicates that 
    \begin{equation}\label{eqn:proof_expectation_eq4}
        \mathbb{E}\widetilde{P}(x) = \left( \prod_{j=1}^d  \mathbb{E}\cos e_j \right) P(x). 
    \end{equation}

    Now let us consider the multi-qubit noisy QSP circuit $\widetilde{U}$. 
    Notice that a single sample of $\widetilde{U}$ is exactly a QSP circuit with phase factors $\widetilde{\Phi}$, then according to~\cref{lem:qubitization},  each sample of $\widetilde{U}$ is a block-encoding of $\widetilde{P}(A)$, i.e., 
    \begin{equation}
        \widetilde{U} = \left( \begin{array}{cc}
            \widetilde{P}(A) & \cdot \\
            \cdot & \cdot
        \end{array} \right). 
    \end{equation}
    By taking the expectation and using~\cref{eqn:proof_expectation_eq4}, we obtain 
    \begin{equation}
        \mathbb{E} \widetilde{U} = \left( \begin{array}{cc}
            \mathbb{E} \widetilde{P}(A) & \cdot \\
            \cdot & \cdot
        \end{array} \right) = \left( \begin{array}{cc}
            \left( \prod_{j=1}^d \mathbb{E} \cos e_j \right) P(A) & \cdot \\
            \cdot & \cdot
        \end{array} \right). 
    \end{equation}
\end{proof}

\section{Algorithms}\label{sec:algorithm}

\cref{thm:QSP_noise_phase_shift} shows that the expectation of a noisy QSP circuit has a top-left sub-block to be a rescaling of $P(A)$. 
Notice that, although each sample of the noisy circuit is a block-encoding of $\widetilde{P}(A)$, the expected circuit is not necessarily a block-encoding of $P(A)$, as the expectation is not necessarily a unitary matrix. 
Nevertheless, we may still use~\cref{thm:QSP_noise_phase_shift} as a theoretical foundation to design quantum algorithm that can mitigate the effect of  errors in QSP by taking the ensemble average over measurements.  

Suppose that $p(x)$ is real polynomial satisfying the assumptions in~\cref{lem:QSP}, $\ket{\psi}$ is a known quantum state, and $A$ is a Hermitian matrix with $\|A\| \leq 1$. 
We assume access to the block-encoding $U_A$ of $A$, and the state preparation oracle $O_{\psi}$ of $\ket{\psi}$. 
In this section, we will discuss two quantum algorithms for $p(A)$. 
The first algorithm constructs a block-encoding of $p(A)$, and the second algorithm estimates the observable $\braket{\psi|p(A)^{\dagger}O p(A)|\psi}$ for a Hermitian matrix $O$.  

\subsection{Block-encoding polynomials}

To construct a block-encoding of $p(A)$, inspired by~\cref{thm:QSP_noise_phase_shift}, we may independently construct multiple samples of the block-encoding of $\widetilde{p}(A) = \frac{1}{2} (\widetilde{P}(A)+\widetilde{P}(A)^{\dagger})$ and then linearly combine them using the quantum linear combination of unitaries (LCU) technique~\cite{ChildsWiebe2012}. 
Detailed steps are as follows. 

\renewcommand{\tablename}{Box}
\begin{table}[htbp]
    \caption{Algorithm for constructing a block-encoding of $p(A)$ by noisy QSP}
  \renewcommand{\arraystretch}{1.2}
  \begin{tabular}{*{1}{@{}L{18cm}}}
    \toprule
    {\bfseries Algorithm \thetable} \quad Constructing a block-encoding of $p(A)$ by noisy QSP \tabularnewline
    \bottomrule
    \textbf{Input:} Block-encoding $U_A$ of $A$, a positive integer $M$ \tabularnewline
    \textbf{for} $m = 1,2,\cdots,M$ \textbf{do} \tabularnewline
    \quad Implement block-encodings $\widetilde{U}_m$ of $\widetilde{P}_{2m-1}(A)$, and $\widetilde{V}_m$ of $\widetilde{P}_{2m}(A)^{\dagger}$, independently via noisy QSP, where each $\widetilde{P}_{k}(x)$ is a noisy sample of $P(x)$ and $P(x)$ is the complex polynomial associated with $p(x)$ as in~\cref{lem:QSP}. \tabularnewline
    \textbf{end for}  \tabularnewline
    Construct the operator 
    \begin{equation}\label{eqn:LCU_op}
        \left(\mathrm{H}^{\otimes \log_2(2M)} \otimes I \right) \left( \sum_{m=1}^{M} \left(\ket{2m-2}\bra{2m-2} \otimes \widetilde{U}_m + \ket{2m-1}\bra{2m-1} \otimes \widetilde{V}_m \right)\right) \left(\mathrm{H}^{\otimes \log_2(2M)} \otimes I \right). 
    \end{equation} \tabularnewline
    \textbf{Output:} Block-encoding of $\widetilde{p}(A)$
    \\ \bottomrule
    \end{tabular}
    \label{alg:block_encoding}
\end{table}

\begin{thm}\label{thm:noisy_QSP_block_encoding}
    Let $U_A$ be a block-encoding of a Hermitian matrix $A$. 
    Then,~Algorithm~\ref{alg:block_encoding} constructs a block-encoding of $\left(\prod_{j=1}^d \mathbb{E} \cos e_j\right)(p(A)+E) $ with probability at least $1-\delta$ such that $\|E\| \leq \epsilon$, by choosing
        \begin{equation}
             M = \mathcal{O}\left( \frac{1}{\epsilon^2 \left(\prod_{j=1}^d \mathbb{E} \cos e_j\right)^2 } \log\left(\frac{1}{\delta}\right) \right),  
        \end{equation}
    and using
        \begin{equation}
            \mathcal{O}\left( \frac{d}{\epsilon^2 \left(\prod_{j=1}^d \mathbb{E} \cos e_j\right)^2 } \log\left(\frac{1}{\delta}\right) \right)
        \end{equation}
        queries to $U_A$. 
\end{thm}

\begin{proof}
    Let $\widetilde{W}$ denote the operator in~\cref{eqn:LCU_op}. 
    Notice that $\widetilde{U}_m$ and $\widetilde{V}_m$ are block-encodings of $\widetilde{P}_{2m-1}(A)$ and $\widetilde{P}_{2m}(A)^{\dagger}$, respectively. 
    According to~\cite[Lemma 52]{GilyenSuLowEtAl2019}, $\widetilde{W}$ is a block-encoding of $\frac{1}{2M}\sum_{m=1}^M (\widetilde{P}_{2m-1}(A) + \widetilde{P}_{2m}(A)^{\dagger})$, i.e., 
    \begin{equation}
        \left(\bra{0}\otimes \bra{0}\otimes I\right) \widetilde{W} \left(\ket{0}\otimes \ket{0} \otimes I\right) = \frac{1}{2M}\sum_{m=1}^M (\widetilde{P}_{2m-1}(A) + \widetilde{P}_{2m}(A)^{\dagger}). 
    \end{equation}
    Here the first $\ket{0}$ represents the LCU ancilla register, and the second $\ket{0}$ represents the QSP ancilla register. 
    Let 
    \begin{equation}
        \widetilde{\mathcal{P}} = \frac{1}{2M}\sum_{m=1}^M (\widetilde{P}_{2m-1}(A) + \widetilde{P}_{2m}(A)^{\dagger}), 
    \end{equation}
    so $\widetilde{W}$ block-encodes $\widetilde{\mathcal{P}}$. 

    We now show that $\widetilde{\mathcal{P}}$ is close to $\left(\prod_{j=1}^d \mathbb{E} \cos e_j\right) p(A)$ with high probability. 
    According to~\cref{lem:qubitization}, we have 
    \begin{equation}
        \mathbb{E} \widetilde{\mathcal{P}} = \frac{1}{2M} \sum_{m=1}^M  \left(\prod_{j=1}^d \mathbb{E} \cos e_j\right) \left(  P(A) + P(A)^{\dagger} \right) = \left(\prod_{j=1}^d \mathbb{E} \cos e_j\right) p(A). 
    \end{equation}
    Notice that each $\widetilde{P}_k$ is bounded within $[-1,1]$ as it is constructed by a QSP circuit. 
    Hoeffding's inequality implies that, for $\epsilon' > 0$, 
    \begin{align}
        \mathbb{P} \left( \left| \widetilde{\mathcal{P}} - \left(\prod_{j=1}^d \mathbb{E} \cos e_j\right) p(A) \right| \geq \epsilon' \right) &\leq 2 \exp\left( - \frac{2 (2M\epsilon')^2 }{ 2M (1-(-1))^2 }  \right) \\
        &= 2 \exp\left( - M \epsilon'^2  \right). 
    \end{align}
    To bound this by $\delta$, it suffices to choose 
    \begin{equation}
        M = \mathcal{O}\left( \frac{1}{\epsilon'^2} \log\left(\frac{1}{\delta}\right) \right). 
    \end{equation}
    We complete the proof by choosing $\epsilon' = \left(\prod_{j=1}^d \mathbb{E} \cos e_j\right) \epsilon$ and noticing that query complexity to $U_A$ of LCU is $\mathcal{O}(Md)$. 
\end{proof}

\subsection{Estimating observables}

Now we consider estimating the observable $\braket{\psi|p(A)^{\dagger}Op(A)|\psi}$ for a known state $\ket{\psi}$ and a Hermitian matrix $O$ with $\|O\|\leq 1$. 
Notice that, using the block-encoding of $p(A)$ we construct in the previous subsection, we can directly estimate the observable using the Hadamard test for non-unitary matrix~\cite{TongAnWiebe2021}. 
However, constructing this block-encoding requires coherent implementation of multiple QSP circuits and thus has deep circuit depth. 
To reduce the circuit depth, we shall consider the importance sampling technique to design a hybrid quantum-classical algorithm for estimating the observable. 
To this end, we first discuss a variant of Hadamard test for estimating random observable.

\subsubsection{Hadamard test for non-unitary random matrices}

Let $\widetilde{O}$ be a random matrix with $\|\widetilde{O}\|\leq 1$. 
Suppose that we are given a block-encoding $U_{\widetilde{O}}$ of $\widetilde{O}$ and that $\mathbb{E}\widetilde{O}$ is a Hermitian matrix. 
Notice that the unitary $U_{\widetilde{O}}$ is also a random matrix. 
The following algorithm estimates the observable $\braket{\psi|\mathbb{E}\widetilde{O}|\psi}$.

\begin{figure}
\centering
    \includegraphics[width=0.4\linewidth]{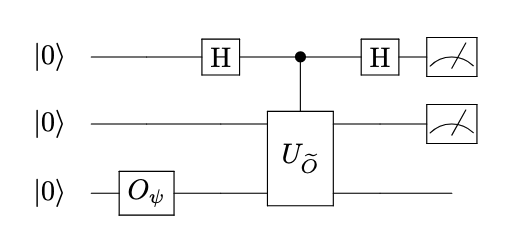}
    \caption{Quantum circuit for estimating $\braket{\psi|\mathbb{E}\widetilde{O}|\psi}$ of a random matrix $\widetilde{O}$ with Hermitian expectation. 
    Here $O_{\psi}$ is the state preparation oracle of $\ket{\psi}$, $U_{\widetilde{O}}$ is a block-encoding of $\widetilde{O}$, and $\mathrm{H}$ is the Hadamard gate. }
    \label{fig:random_Hadamard}
\end{figure}

\begin{table}[htbp]
    \caption{Algorithm for Hadamard test for non-unitary random matrices}
  \renewcommand{\arraystretch}{1.2}
  \begin{tabular}{*{1}{@{}L{18cm}}}
    \toprule
    {\bfseries Algorithm \thetable} \quad Hadamard test for non-unitary random matrices \tabularnewline
    \bottomrule
    \textbf{Input:} Block-encoding $U_{\widetilde{O}}$ of a random matrix $\widetilde{O}$, state preparation oracle $O_{\psi}$ of $\ket{\psi}$,  a positive integer $M$ \tabularnewline
    \textbf{for} $m = 1,2,\cdots,M$ \textbf{do} \tabularnewline
    \quad Implement the circuit in~\cref{fig:random_Hadamard} and construct a random variable $\widetilde{o}_m$ as 
    \begin{equation}
        \widetilde{o}_m = \begin{cases}
            1, & \text{if the measurement outcome is }\ket{0}\ket{0}, \\
            -1, & \text{if the measurement outcome is }\ket{1}\ket{0}, \\
            0, & \text{else}. 
        \end{cases}
    \end{equation} \tabularnewline
    \textbf{end for}  \tabularnewline
    Compute $\frac{1}{M} \sum_{m=1}^M \widetilde{o}_m$ as the estimator of the observable. \tabularnewline
    \textbf{Output:} An estimate of the observable $\braket{\psi|\mathbb{E}\widetilde{O}|\psi}$
    \\ \bottomrule
    \end{tabular}
    \label{alg:random_Hadamard}
\end{table}

\begin{thm}\label{thm:random_Hadamard}
    Let $\ket{\psi}$ be a quantum state and $\widetilde{O}$ be a random matrix with $\|\widetilde{O}\|\leq 1$ whose expectation value is Hermitian. 
    Suppose that we are given access to the state preparation oracle $O_{\psi}$ of $\ket{\psi}$ and a random block-encoding $U_{\widetilde{O}}$ of $\widetilde{O}$. 
    Then,~Algorithm~\ref{alg:random_Hadamard} outputs an $\epsilon$-approximation of $\braket{\psi|\mathbb{E}\widetilde{O}|\psi}$ with probability at least $1-\delta$, using queries to $O_{\psi}$ and $U_{\widetilde{O}}$ for $M$ times where 
    \begin{equation}
        M = \mathcal{O} \left( \frac{1}{\epsilon^2} \log\left(\frac{1}{\delta}\right) \right) . 
    \end{equation}
\end{thm}

\begin{proof}
    We first study the expectation of $\widetilde{o}$. 
    For a fixed sample of the circuit in~\cref{fig:random_Hadamard}, after applying the state preparation oracle $O_{\psi}$ and the first Hadamard gate, we obtain the state $\frac{1}{\sqrt{2}}(\ket{0}+\ket{1})\ket{0}\ket{\psi}$. 
    Applying the controlled block-encoding $U_{\widetilde{O}}$ yields 
    \begin{equation}
        \frac{1}{\sqrt{2}}\left(\ket{0} \ket{0}\ket{\psi} + \ket{1} \ket{0}\widetilde{O}\ket{\psi} \right) + \ket{1}\ket{\perp}, 
    \end{equation}
    where $\ket{\perp}$ is a possibly unnormalized vector such that $(\ket{0}\bra{0}\otimes I) \ket{\perp} = 0$. 
    We then apply the second Hadamard gate and obtain 
    \begin{equation}
        \frac{1}{2} \left(  \ket{0}\ket{0} (I+\widetilde{O})\ket{\psi} + \ket{1}\ket{0} (I-\widetilde{O})\ket{\psi} \right) + (\ket{0}-\ket{1})\ket{\perp}. 
    \end{equation}
    So we have 
    \begin{equation}
        \mathbb{P} \left( \widetilde{o} = 1 | \widetilde{O} \right) = \frac{1}{4} \left\| (I+\widetilde{O})\ket{\psi} \right\|^2 = \frac{1}{2} + \frac{1}{4}\braket{\psi| (\widetilde{O} + \widetilde{O}^{\dagger}) |\psi}, 
    \end{equation}
    \begin{equation}
        \mathbb{P} \left( \widetilde{o} = -1 | \widetilde{O} \right) = \frac{1}{4} \left\| (I-\widetilde{O})\ket{\psi} \right\|^2 = \frac{1}{2} - \frac{1}{4}\braket{\psi| (\widetilde{O} + \widetilde{O}^{\dagger}) |\psi}, 
    \end{equation}
    and thus 
    \begin{equation}
        \mathbb{E} (\widetilde{o}|\widetilde{O}) = \frac{1}{2} \braket{\psi| (\widetilde{O} + \widetilde{O}^{\dagger}) |\psi}. 
    \end{equation}
    By the law of total expectation and noticing that $\mathbb{E}(\widetilde{O}^{\dagger}) = (\mathbb{E}\widetilde{O})^{\dagger} = \mathbb{E}\widetilde{O}$, we have 
    \begin{equation}
        \mathbb{E} (\widetilde{o}) =  \mathbb{E}\left(\mathbb{E} (\widetilde{o}|\widetilde{O})\right) = \mathbb{E} \left( \frac{1}{2} \braket{\psi| (\widetilde{O} + \widetilde{O}^{\dagger}) |\psi} \right) = \braket{\psi|\mathbb{E}\widetilde{O}|\psi}, 
    \end{equation}
    which shows that $\widetilde{o}$ is an unbiased estimator of the desired observable. 

    We now estimate the number of samples needed to estimate the observable with high probability. 
    Notice that all the samples $\widetilde{o}_m$ are always bounded in $[-1,1]$. 
    Hoeffding's inequality implies that 
    \begin{align}
        \mathbb{P}\left( \left|\frac{1}{M} \sum_{m=1}^M \widetilde{o}_m - \braket{\psi|\mathbb{E}(\widetilde{O})|\psi}\right| \geq \epsilon \right)  & \leq 2 \exp \left( - \frac{2\epsilon^2}{\sum_{m=1}^M (2/M)^2 } \right) \\
        & = 2 \exp \left( - \frac{M\epsilon^2}{ 2 } \right). 
    \end{align}
    The choice of $M$ can then be derived by bounding the right hand side by $\delta$. 
\end{proof}

\subsubsection{Estimating observables of noisy QSP}

We now show how to use~Algorithm~\ref{alg:random_Hadamard} to estimate the observable $\braket{\psi|p(A)^{\dagger}O p(A)|\psi}$. 
Notice that we can rewrite the desired observable as 
\begin{equation}
    \braket{\psi|p(A)^{\dagger}O p(A)|\psi} = \frac{1}{4}\braket{\psi|(P(A)+P(A)^{\dagger})O (P(A)+P(A)^{\dagger}) |\psi}. 
\end{equation}

\begin{table}[htbp]
    \caption{Algorithm for estimating the observable by noisy QSP}
  \renewcommand{\arraystretch}{1.2}
  \begin{tabular}{*{1}{@{}L{18cm}}}
    \toprule
    {\bfseries Algorithm \thetable} \quad Estimating the observable by noisy QSP \tabularnewline
    \bottomrule
    \textbf{Input:} Block-encoding $U_A$ of $A$, block-encoding $U_{O}$ of a Hermitian matrix $O$, state preparation oracle $O_{\psi}$ of $\ket{\psi}$ \tabularnewline
    \textbf{for} $k = 1,2,3,4$ \textbf{do} \tabularnewline
    \quad Construct block-encodings of $\widetilde{P}_k(A)$ independently by noisy QSP, where each $\widetilde{P}_k(A)$ is a noisy sample of $P(x)$ and $P(x)$ is the complex polynomial associated with $p(x)$ as in~\cref{lem:QSP}.  \tabularnewline
    \textbf{end for}  \tabularnewline
    Construct a block-encoding $U_{\widetilde{O}}$ of $ \widetilde{O} = \frac{1}{4} \left( \widetilde{P}_1(A) + \widetilde{P}_2(A)^{\dagger} \right) O \left( \widetilde{P}_3(A) + \widetilde{P}_4(A)^{\dagger} \right)$.  \tabularnewline
    Estimate $o = \braket{\psi|\mathbb{E}\widetilde{O}|\psi}$ using~Algorithm~\ref{alg:random_Hadamard}.  \tabularnewline
    Compute $o / \left( \prod_{j=1}^d \cos e_j \right)^2$.  \tabularnewline
    \textbf{Output:} An estimate of the observable $\braket{\psi|p(A)^{\dagger}O p(A)|\psi}$
    \\ \bottomrule
    \end{tabular}
    \label{alg:observable}
\end{table}

\begin{thm}\label{thm:noisy_QSP_observable}
    Let $p(x)$ be a $d$-degree real polynomial satisfying the assumptions in~\cref{lem:QSP}, $A$ and $O$ be Hermitian matrices with $\|A\| \leq 1$ and $\|O\|\leq 1$, and $\ket{\psi}$ be a quantum state. 
    Suppose that we are given access to a block-encoding $U_A$ of $A$, block-encoding $U_O$ of $O$, and a state preparation oracle $O_{\psi}$ of $\ket{\psi}$. 
    Then,~Algorithm~\ref{alg:observable} outputs an $\epsilon$-approximation of $\braket{\psi|p(A)^{\dagger}O p(A)|\psi}$ with probability at least $1-\delta$. 
    Furthermore, 
    \begin{enumerate}
        \item the algorithm independently runs 
        \begin{equation}
            \mathcal{O}\left( \frac{1}{\epsilon^2 \left( \prod_{j=1}^d \mathbb{E} \cos e_j \right)^4 } \log\left(\frac{1}{\delta}\right)\right)
        \end{equation}
        many quantum circuits, 
        \item each circuit uses $\mathcal{O}(d)$ queries to $U_A$, and $\mathcal{O}(1)$ queries to $U_O$ and $O_{\psi}$. 
    \end{enumerate}
\end{thm}

\begin{proof}
    For any $\widetilde{P}_k(A)$, according to~\cref{thm:QSP_noise_phase_shift}, we have $\mathbb{E} \widetilde{P}_k(A) = \left( \prod_{j=1}^d \mathbb{E} \cos e_j \right) P(A)$. 
    Using the independence among different $\widetilde{P}_k(A)$, we have 
    \begin{align}
        \mathbb{E} \widetilde{O} &= \frac{1}{4} \left( \mathbb{E}\widetilde{P}_1(A) + \mathbb{E}\widetilde{P}_2(A)^{\dagger} \right) O \left( \mathbb{E}\widetilde{P}_3(A) + \mathbb{E}\widetilde{P}_4(A)^{\dagger} \right) \\
        & = \frac{1}{4} \left( \prod_{j=1}^d \mathbb{E} \cos e_j \right)^2 \left( P(A) + P(A)^{\dagger} \right) O \left( P(A) + P(A)^{\dagger} \right) \\
        & = \left( \prod_{j=1}^d \mathbb{E} \cos e_j \right)^2 p(A)^{\dagger} O p(A). 
    \end{align}
    So $o / \left( \prod_{j=1}^d \cos e_j \right)^2 $ is an $\epsilon$-approximation of the desired observable, if $o$ approximates $\braket{\psi|\mathbb{E}\widetilde{O}|\psi}$ up to error $\left( \prod_{j=1}^d \cos e_j \right)^2\epsilon$. 
    Then the claimed query complexities can be derived by~\cref{thm:random_Hadamard}. 
\end{proof}

\section{Applications}\label{sec:applications}

We apply our algorithms to three specific applications: simulating a time-independent Hamiltonian evolution, solving a linear system of equations, and preparing the ground state of a Hamiltonian.
For technical clarity, in this section, we assume in addition to the conditions in~\cref{sec:error_model}, that all the random noises $e_j$ are also identically distributed. 
For all $j$, we denote 
\begin{equation}
    c = \mathbb{E} \cos e_j \approx 1-\frac{1}{2}\nu, 
\end{equation}
where $\nu = \mathbb{E}(e^2)$ is the variance of the error.

\subsection{Hamiltonian simulation}

Let $H$ be a Hamiltonian, and we are given a block-encoding $U_H$ of $H/\|H\|$. 
The goal of the Hamiltonian simulation problem is to construct a block-encoding of $e^{-iHT}$, or prepare a quantum state $e^{-iHT}\ket{\psi}$ for a known input state $\ket{\psi}$, or estimating the observable $\braket{\psi|e^{iHT} O e^{-iHT} |\psi}$. 

Notice that 
\begin{equation}
    e^{-iHT} = \cos(HT) - i \sin (HT). 
\end{equation}
So we may separately implement $\cos(HT)$ and $\sin (HT)$ by their polynomial approximation established in~\cite[Lemma 57]{GilyenSuLowEtAl2019}, and then linearly combine them together to obtain $e^{-iHT}$. 
We summarize this result in the following lemma. 
\begin{lem}[{\cite[Lemma 57 and 59]{GilyenSuLowEtAl2019}}]\label{lem:Hsim_polynomial_approx}
    Let $\beta > 0$ and $\epsilon \in (0,1/e)$. 
    Then there exist two $d$-degree polynomials $p_1(x)$ and $p_2(x)$ with parity $d\text{ mod } 2$ such that 
    \begin{equation}
        \max_{x\in[-1,1]} |\cos(\beta x) - p_1(x)| \leq \epsilon, \quad \max_{x\in[-1,1]} |\sin(\beta x) - p_2(x)| \leq \epsilon. 
    \end{equation}
    Furthermore, for any positive real number $q > 0$, the degree of the polynomials can be bounded as 
    \begin{equation}
        d \leq e^{q+1} \beta/2 + \log(4/(5\epsilon))/q + 1. 
    \end{equation}
\end{lem}

In the complexity estimates of our algorithms (\cref{thm:noisy_QSP_block_encoding} and~\cref{thm:noisy_QSP_observable}), we encounter a rescaling factor $\prod_{j=1}^d \mathbb{E} \cos e_j$. It can be estimated as
\begin{align}\label{eqn:cd}
    \prod_{j=1}^d \mathbb{E} \cos e_j =c^d = (1-\frac{1}{2}\nu)^d \approx e^{-\frac{1}{2}d\nu}.
\end{align}
When the system noise is sufficiently small, i.e., $d\nu \ll \Or(1)$, the factor is independent of the system size.

We are now ready to analyze the complexity of using noisy QSP for the Hamiltonian simulation problem. 

\begin{thm}\label{thm:app_Hsim}
    Let $H$ be a Hamiltonian, $O$ be a Hermitian matrix with $\|O\|\leq 1$, and $\ket{\psi}$ be a quantum state. 
    Suppose that we are given a block-encoding $U_H$ of $H/\|H\|$, a block-encoding $U_O$ of $O$, and a state preparation oracle $O_{\psi}$ of $\ket{\psi}$. 
    Consider using noisy QSP to implement $e^{-iHT}$ with i.i.d. random additive noises $e_j$ in the phase factors satisfying $c = \mathbb{E} \cos e_j$.  
    Let $q$ be an arbitrary real positive constant, then  
    \begin{enumerate}
        \item we can construct a block-encoding of $\frac{1}{\alpha}(e^{-iHT} + E)$ with $\|E\| \leq \epsilon$, success probability at least $1-\delta$, and a block-encoding factor 
        \begin{equation}
            \frac{1}{\alpha} \approx \frac{1}{4}e^{-\frac{1}{2}d\nu}, 
        \end{equation}
        using 
        \begin{equation}
            \mathcal{O}\left( e^{d\nu}\left(\frac{1}{\epsilon}\right)^{2} \left( \|H\|T + \log\left(\frac{1}{\epsilon}\right) \right)\log\left(\frac{1}{\delta}\right) \right)
        \end{equation}
        queries to $U_H$, 
        \item the quantum state $e^{-iHT}\ket{\psi}$ can be prepared up to error $2\epsilon$ with probability at least $1-\delta$, using 
        \begin{equation}
          \mathcal{O}\left( e^{2d\nu}\left(\frac{1}{\epsilon}\right)^{2} \left( \|H\|T + \log\left(\frac{1}{\epsilon}\right) \right)\left(\log\left(\frac{1}{\delta}\right)\right)^2 \right)
        \end{equation}
        queries to $U_H$, and 
        \begin{equation}
        \mathcal{O}\left( e^{d\nu} \log\left(\frac{1}{\delta}\right) \right)
        \end{equation}
        queries to $O_{\psi}$, 
        \item the observable $\braket{\psi|e^{iHT} O e^{-iHT} |\psi}$ can be estimated up to error $\epsilon$ with probability at least $1-\delta$, using  
        \begin{equation}
            \mathcal{O}\left( e^{2d\nu}\left(\frac{1}{\epsilon}\right)^{2} \log\left(\frac{1}{\delta}\right) \right)
        \end{equation}
        independent runs of a quantum circuit, each of which takes $\mathcal{O}(\|H\|T + \log(1/\epsilon))$ queries to $U_H$, and $\mathcal{O}(1)$ queries to $U_O$ and $O_{\psi}$. 
    \end{enumerate}
\end{thm}

\begin{proof}
    We first start with constructing a block-encoding of $\cos(HT)$. 
    \cref{lem:Hsim_polynomial_approx} shows that for any positive real number $q$, there exists a real polynomial $p_1(x)$ with degree 
    \begin{equation}
        d \leq e^{q+1} \beta/2 + \log(8/(5\epsilon))/q + 1
    \end{equation}
    such that 
    \begin{equation}
        \max_{x \in [-1,1]}| \cos(\beta x)/2 - p_1(x) | \leq \epsilon/4. 
    \end{equation}
    By choosing $\beta = \|H\|T$ and using~\cref{thm:noisy_QSP_block_encoding}, a block-encoding of $\frac{c^d}{2} (\cos(HT) + E_1)$ can be constructed with $\|E_1\| \leq \epsilon/2$ and probability at least $1-\delta/2$, using 
    \begin{equation}\label{eqn:proof_Hsim_eq1}
        \mathcal{O}\left( \frac{d}{\epsilon^2 c^{2d} } \log\left(\frac{1}{\delta}\right) \right). 
    \end{equation}
    queries to $U_H$. 
    Similarly, a block-encoding of $\frac{c^d}{2} (\sin(HT) + E_2)$ can be constructed with the same cost and error. 
    Using LCU again, we may construct a block-encoding of $\frac{c^d}{4} (\cos(HT) - i\sin(HT) + E_1 - i E_2) $ with one query to the block-encodings of $\frac{c^d}{2} (\cos(HT) + E_1)$ and $\frac{c^d}{2} (\sin(HT) + E_2)$. 
    This is in turn a block-encoding of $\frac{c^d}{4} (e^{-iHT} + E)$ with $\|E\| \leq \epsilon$, success probability at least $1-\delta$, and query complexity as in~\cref{eqn:proof_Hsim_eq1}. 
    The scalings claimed in the first part of the theorem can then be directly derived from~\eqref{eqn:cd}.

    To prepare the quantum state $e^{-iHT}\ket{\psi}$, we may simply apply the block-encoding on the input state $\ket{\psi}$ and measure the ancilla register onto $0$. 
    Specifically, after applying the block-encoding of $\frac{c^d}{4} (e^{-iHT} + E)$, we obtain the state 
    \begin{equation}
        \frac{c^d}{4} \ket{0}  (e^{-iHT} + E) \ket{\psi} + \ket{\perp}
    \end{equation}
    where $\ket{\perp}$ is a possibly unnormalized orthogonal state. 
    Upon measuring the ancilla register onto $0$, we obtain the state $(e^{-iHT} + E) \ket{\psi} / \|(e^{-iHT} + E) \ket{\psi}\|$. 
    To bound the distance between this state and the ideal state $e^{-iHT}\ket{\psi}$, we use the linear algebra result for two possibly unnormalized vectors $\vec{a}$ and $\vec{b}$ that 
    \begin{align}
        \left\|\frac{\vec{a}}{\|\vec{a}\|} - \frac{\vec{b}}{\|\vec{b}\|}\right\| &= \frac{\left\|\|\vec{b}\|\vec{a} -  \|\vec{a}\|\vec{b} \right\|}{\|\vec{a}\|\|\vec{b}\|} \\
        &= \frac{1}{\|\vec{a}\|\|\vec{b}\|} \left\| \|\vec{b}\|\vec{a} - \|\vec{a}\|\vec{a} + \|\vec{a}\|\vec{a}  - \|\vec{a}\|\vec{b}\right\| \\
        & \leq \frac{1}{\|\vec{a}\|\|\vec{b}\|} \left(\left\| \|\vec{b}\|\vec{a} - \|\vec{a}\|\vec{a}\right\| + \left\| \|\vec{a}\|\vec{a}  - \|\vec{a}\|\vec{b}\right\|\right)  \\
        & = \frac{1 }{\|\vec{b}\|} \left(  \left\| \|\vec{b}\| - \|\vec{a}\|\right\| +  \left\| \vec{a}  - \vec{b}\right\| \right) \\
        & \leq \frac{2}{\|\vec{b}\|} \left\| \vec{a}  - \vec{b}\right\|. 
    \end{align}
    Let $\vec{a} = (e^{-iHT} + E) \ket{\psi}$ and $\vec{b} = e^{-iHT}\ket{\psi}$, we have 
    \begin{equation}
        \left\| \frac{(e^{-iHT} + E) \ket{\psi}}{\|(e^{-iHT} + E) \ket{\psi}\|} - e^{-iHT}\ket{\psi} \right\| \leq 2 \|E\ket{\psi}\| \leq 2\epsilon.  
    \end{equation}
    The success probability per run is 
    \begin{equation}
        \frac{c^{2d}}{16} \|(e^{-iHT} + E)\ket{\psi}\|^2 (1-\delta) \geq \frac{c^{2d}}{128}. 
    \end{equation}
    If we run the algorithm for $K$ times, then the probability that at least one run is successful becomes at least 
    \begin{equation}
        1 - \left(1-\frac{c^{2d}}{128}\right)^K. 
    \end{equation}
    In order to bound this from below by $1-\delta$, it suffices to choose 
    \begin{equation}
        K = \mathcal{O}\left( \frac{1}{\log\left(\frac{1}{1-c^{2d}/32}\right)} \log\left(\frac{1}{\delta}\right) \right) = \mathcal{O}\left( \frac{1}{c^{2d} } \log\left(\frac{1}{\delta}\right) \right). 
    \end{equation}
    The overall complexity is $K$ times the cost of a single run, and by plugging in the estimate of $c^d$ in~\eqref{eqn:cd} we obtain the scalings as claimed. 

    For estimating the observable $\braket{\psi|e^{iHT} O e^{-iHT} |\psi}$, let $p_1(x)$ and $p_2(x)$ be the real polynomials that approximate $\cos(\|H\| T x)/2$ and $\sin(\|H\| T x)/2$ up to error $\epsilon'$, respectively, as in~\cref{lem:Hsim_polynomial_approx}. 
    The degree of these two polynomials are 
    \begin{equation}
        d \leq e^{q+1} \|H\| T /2 + \log(2/(5\epsilon'))/q + 1. 
    \end{equation}
    We write 
    \begin{equation}
        \frac{1}{2}\cos(HT) = \frac{1}{2} p_1(H/\|H\|) + E_1, \quad \frac{1}{2}\sin(HT) = \frac{1}{2} p_2(H/\|H\|) + E_2, 
    \end{equation}
    where $\|E_1\|$ and $\|E_2\|$ are bounded by $\epsilon'$. 
    Then 
    \begin{align}
        \braket{\psi|e^{iHT} O e^{-iHT} |\psi} &= \braket{\psi|\cos(HT) O \cos(HT) |\psi} +  \braket{\psi|\sin(HT) O \sin(HT) |\psi} \\
        & = \braket{\psi|p_1(H/\|H\|) O p_1(H/\|H\|) |\psi} +  \braket{\psi|p_2(H/\|H\|) O p_2(H/\|H\|) |\psi} + E, 
    \end{align}
    where 
    \begin{equation}
        E = 4\left( \braket{\psi|E_1 O \cos(HT) |\psi} + \braket{\psi|E_1 O E_1 |\psi} + \braket{\psi|E_2 O \sin(HT) |\psi} + \braket{\psi|E_2 O E_2 |\psi} \right), 
    \end{equation}
    and 
    \begin{equation}
        \|E\| \leq 8 (\epsilon' + \epsilon'^2). 
    \end{equation}
    Therefore, we may estimate the desired observable $\braket{\psi|e^{iHT} O e^{-iHT} |\psi} $ up to error $\epsilon$ by choosing $\epsilon' = \epsilon/32$ and estimating both $\braket{\psi|p_1(H/\|H\|) O p_1(H/\|H\|) |\psi}$ and $\braket{\psi|p_2(H/\|H\|) O p_2(H/\|H\|) |\psi}$ up to error $\epsilon/4$. 
    The claimed complexity directly follows from~\cref{thm:noisy_QSP_observable} and~\eqref{eqn:cd}.  
\end{proof}

\subsection{Solving linear systems of equations}

The goal of a quantum linear system problem (QLSP) is to prepare a normalized solution of $Ax = b$, i.e., 
\begin{equation}
    \ket{x} = \frac{A^{-1} \ket{b}}{ \|A^{-1}\ket{b}\|}, 
\end{equation}
or estimating the corresponding observables. 
Since QLSP only needs normalized solution, we may assume $\|A\| = \|b\| = 1$ without loss of generality, and assume access to the block-encoding $U_A$ of $A$ and state preparation oracle $O_b$ of $\ket{b}$. 
Furthermore, we only focus on the case where $A$ is a Hermitian matrix. 
The general case can be reduced to the Hermitian one by the standard dilation trick~\cite{HarrowHassidimLloyd2009}. 

Note that~\cite[Corollary 69]{GilyenSuLowEtAl2019} gives a polynomial approximation of the function $1/x$ over an interval away from $0$. 
Again the noisy QSP may implement this polynomial approximation solve the QLSP. 
We study the complexity in terms of the tolerated error $\epsilon$ and the condition number $\kappa = \|A\|\|A^{-1}\|$. 

\begin{thm}\label{thm:app_QLSP}
    Let $A$ be a Hermitian matrix with $\|A\| = 1$ and condition number $\kappa$, $O$ be a Hermitian matrix with $\|O\| \leq 1$, and $\ket{b}$ be a quantum state. 
    Suppose that we are given a block-encoding $U_A$ of $A$, a block-encoding $U_O$ of $O$, and a state preparation oracle $O_b$ of $\ket{b}$. 
    Consider the noisy QSP algorithms for the linear system problem $Ax = b$ with i.i.d. random additive noises $e_j$ in the phase factors satisfying $c = \mathbb{E} \cos e_j$. 
    Then 
    \begin{enumerate}
        \item a block-encoding of $\frac{1}{\alpha}(A^{-1} + E)$ can be constructed with error $\|E\| \leq \epsilon$, success probability at least $1-\delta$, and a block-encoding factor 
        \begin{equation}
            \frac{1}{\alpha} \approx \frac{3}{4\kappa}e^{-\frac{1}{2}d\nu},
        \end{equation}
        using 
        \begin{equation}
            \mathcal{O}\left( e^{d\nu}\left(\frac{1}{\epsilon}\right)^2\kappa^3\log\left(\frac{\kappa}{\epsilon}\right)\log\left(\frac{1}{\delta}\right) \right)
        \end{equation}
        queries to $U_A$. 
        \item the quantum state $\ket{x} = x/\|x\|$ can be prepared with error at most $2\epsilon$ and probability at least $1-\delta$, using 
        \begin{equation}
            \mathcal{O}\left( e^{2d\nu}\left(\frac{1}{\epsilon}\right)^2\kappa^3\log\left(\frac{\kappa}{\epsilon}\right)\left(\log\left(\frac{1}{\delta}\right)\right)^2 \right)
        \end{equation}
        queries to $U_A$, and 
        \begin{equation}
            \mathcal{O}\left( e^{d\nu}\log\left(\frac{1}{\delta}\right) \right) 
        \end{equation}
        queries to $O_b$. 
        \item the observable $x^* O x$ can be estimated with error at most $\epsilon$ and probability at least $1-\delta$, using 
        \begin{equation}
            \mathcal{O}\left( e^{2d\nu}\frac{\kappa^{4}}{\epsilon^{2} } \log\left(\frac{1}{\delta}\right)\right) 
        \end{equation}
        independent runs of a quantum circuit, each of which takes $\mathcal{O}(\kappa\log(\kappa/\epsilon))$ queries to $U_A$, and $\mathcal{O}(1)$ queries to $U_O$ and $O_b$. 
    \end{enumerate}
\end{thm}
\begin{proof}
    \cite[Lemma 69]{GilyenSuLowEtAl2019} shows that there exists an odd $d$-degree real polynomial $p(x)$ such that $\max_{x\in[-1,1]}|p(x)| \leq 1$,  
    \begin{equation}
        \max_{x\in[-1,-1/\kappa]\cup[1/\kappa,1]}\left| \frac{3}{4} \frac{1}{\kappa x} - p(x) \right| \leq \epsilon', 
    \end{equation}
    and 
    \begin{equation}
        d = \mathcal{O}\left( \kappa \log\left(\frac{1}{\epsilon'}\right) \right). 
    \end{equation}
    According to~\cref{thm:noisy_QSP_block_encoding}, we may implement a block-encoding of $c^d (p(A) + E' )$ with error $\|E'\| \leq \epsilon'$ and query complexity 
    \begin{equation}
        \mathcal{O}\left( \frac{d}{\epsilon'^2 c^{2d}} \log\left(\frac{1}{\delta}\right) \right). 
    \end{equation}
    This unitary is in turn a block-encoding of $\frac{3c^d}{4\kappa}\left( A^{-1} + E \right)$ with error $\|E\| \leq 8\kappa \epsilon'/3$. 
    We may choose $\epsilon' = 3\epsilon/(8\kappa)$ to bound this error by $\epsilon$. 
    Similarly, we encounter the factor $c^d$ which can be estimated as in Eq.~\eqref{eqn:cd}. Then we can derive the claimed complexity for constructing block-encoding. 

    Applying this block-encoding to the state $\ket{0}\ket{b}$ yields 
    \begin{equation}
        \frac{3c^d}{4\kappa} \left( \ket{0}(A^{-1}+E)\ket{b} \right) + \ket{\perp}. 
    \end{equation}
    A measurement of the ancilla register onto $0$ gives the state $(A^{-1}+E)\ket{b}/\|(A^{-1}+E)\ket{b}\|$, which is an $2\epsilon$ approximation of $\ket{x}$, and the success probability is $\Omega(\kappa^2/(c^{2d}\|A^{-1}\ket{b}\|^2)) = \Omega(1/c^{2d})$. 
    This can be proved using the same approach as in~\cref{thm:app_Hsim}. 

    For estimating the observable, we need to determine the error level in estimating polynomial observable. 
    Following the notations, we have 
    \begin{equation}
         A^{-1} = \frac{4\kappa}{3} p(A) + \frac{4\kappa}{3} E',  
    \end{equation}
    where $\|E'\| \leq \epsilon'$. 
    Then we have 
    \begin{align}
        x^{*} O x &= \braket{b|A^{-1}OA^{-1}|b} \\
        &= \frac{16\kappa^2}{9} \braket{b|p(A)^{\dagger}Op(A)|b} + E,
    \end{align}
    where 
    \begin{equation}
        E = \frac{16\kappa^2}{9} \left( 2\braket{b|E'Op(A)|b} + \braket{b|E'OE'|b} \right), 
    \end{equation}
    with $\|E\| \leq \mathcal{O}(\kappa^2\epsilon')$. 
    Therefore, in order to bound the observable estimation error by $\epsilon$, it suffices to choose $\epsilon' = \mathcal{O}(\epsilon/\kappa^2)$ and estimate $\braket{b|p(A)^{\dagger}Op(A)|b}$ with error at most $\mathcal{O}(\epsilon/\kappa^2)$. 
    \cref{thm:noisy_QSP_observable} gives the desired complexity. 
\end{proof}

\subsection{Ground state preparation}
Given any positive semidefinite Hermitian matrix $\tilde{H}$, we properly normalize and shift it into the Hamiltonian $H$ with $\|H\| \le 1$, ground state energy $\lambda_0 >0$ and spectral gap $\Delta:=\lambda_1-\lambda_0$. The objective here is to prepare the ground state $\ket{\psi_0}$ of the Hamiltonian $H$, with $\gamma=|\langle \phi | \psi_0 \rangle|$ overlap between the initial state $\ket{\phi}$ and the ground state $\ket{\psi_0}$. Assuming access to the state preparation oracle $O_I$ for $\ket{\phi}$, we employ noisy QSP to approximate the desired filter function and thereby address the problem, following the approach outlined in Ref.\cite{Dong_Lin_Tong_2022}.

Consider the filter function $F(x)$ defined in \cite{Dong_Lin_Tong_2022}, which incorporates a cosine transformation in its argument. By the construction of [\cite{LowChuang2019}, Corollary 7], it can be approximated to accuracy $\epsilon$ by a degree $d=\mathcal{O}(\Delta^{-1}\log(1/\epsilon))$ polynomial. [Theorem 1~\cite{Dong_Lin_Tong_2022}] ensures that any admissible real polynomial $g$ admits a QETU circuit yielding a block-encoding of $g(\cos(H))$ through interleaved $X$-rotations and Hamiltonian simulation oracles. Moreover, as shown in [\cite{Dong_Lin_Tong_2022} Theorem~16], this QETU circuit is equivalent to a QSP circuit obtained by Hadamard conjugation on the ancilla. Consequently, a noisy QSP circuit with access to Hamiltonian simulation suffices to approximate the filter function for solving the problem. We study the complexity in terms of the spectral gap $\Delta$ and the initial overlap $\gamma$.
\begin{thm}\label{thm:app_GSP} Let $H$ be a Hermitian matrix with spectral norm $\|H\|\le 1$, ground state energy $\lambda_0 >0$ and spectral gap $\Delta$, $O$ be a Hermitian matrix with $\|O\| \le 1$. Suppose that we are given a Hamiltonian simulation oracle $O_H$ of $H$, a state preparation oracle $O_I$ for the initial state $\ket{\phi}$ and a block-encoding $U_O$ of $O$. Consider the noisy QSP algorithms for preparing the ground state of $H$ with i.i.d. random additive noises $e_j$ in the phase factors satisfying $c= \mathbb{E}\cos e_j$. Then
\begin{enumerate}
    \item a block-encoding of $\frac{1}{\alpha}(F(\cos H)+E)$ can be constructed with error $\|E\| \le \epsilon$, success probability at least $1-\delta$, and a block-encoding factor
    \begin{align}
        \frac{1}{\alpha} \approx e^{-\frac{1}{2}d\nu},
    \end{align}
    using
   \begin{align}
        \mathcal{O}\left(e^{d\nu}\frac{\Delta^{-1}}{\epsilon^{2}} \log(\frac{1}{\epsilon})\log(\frac{1}{\delta})\right)
    \end{align}
    queries to $O_H$. 
    \item  the ground state $\ket{\psi_0}$ can be prepared with error at most $2\epsilon$ and probability at least $1-\delta$, using 
    \begin{align}
        \mathcal{O}\left(e^{\frac{3}{2}d\nu}\frac{\Delta^{-1}}{\epsilon^{2}\gamma} \log(\frac{1}{\epsilon})\left(\log(\frac{1}{\delta})\right)^2\right)
    \end{align}
    queries to $O_H$ and 
       \begin{equation}
        \mathcal{O}\left(e^{\frac{1}{2}d\nu}\frac{1}{\gamma}\log(\frac{1}{\delta})\right)
    \end{equation}
    queries to $O_I$ with amplitude amplification. 
    \item the observable $\bra{\psi_0}O\ket{\psi_0}$ can be estimated up to error $\epsilon$ with probability at least $1-\delta$, using
    \begin{equation}
        \mathcal{O}\left(e^{2d\nu}\frac{\gamma^{-4}}{\epsilon^{2}} \log(\frac{1}{\delta})\right)
    \end{equation}
    independent runs of a quantum circuit, each of which takes $\mathcal{O}(\Delta^{-1}\log(1/\epsilon\gamma^2))$ queries to $O_H$, and $\mathcal{O}(1)$ queries to $U_O$ and $O_I$. 
\end{enumerate}
\end{thm}

\begin{proof}
Consider the filter function introduced in \cite{Dong_Lin_Tong_2022}, which is a real even polynomial satisfying $\max_{x \in [-1,1]} |F(x)| \le 1$, 
\begin{equation}
\left\{
\begin{aligned}
    & |F(x)-1| \le \epsilon', \quad x\in [\cos(\mu+\tfrac{\Delta}{2}),\cos(\eta)] \\
    & |F(x)| \le \epsilon', \quad x\in [\cos(1-\eta),\cos(\mu-\tfrac{\Delta}{2})]
\end{aligned}
\right.
\end{equation}
where $0<\eta\le \lambda_0 \le \mu-\Delta/2<\mu+\Delta/2\le\lambda_1<1-\eta$. It can be constructed via a polynomial $g(x)$ given in [\cite{LowChuang2019}, Corollary 7] with degree 
\begin{equation}
    d=\mathcal{O}(\Delta^{-1}\log(1/\epsilon)).
\end{equation} According to Theorem \ref{thm:noisy_QSP_block_encoding} while replacing to the block-encoding $p(A)$ with $U_A$ access to the consine block-encoding $g(\cos H)$ with $O_H$ access, we can implement a block encoding of $c^d (g(\cos H)+ E')$ with error $\|E'\| \le \epsilon'$ and query complexity
\begin{equation}
        \mathcal{O}\left( \frac{d}{\epsilon'^2 c^{2d}} \log\left(\frac{1}{\delta}\right) \right). 
    \end{equation}
The unitary is then a block-encoding of $c^d(F(\cos H)+E)$ with error $\|E\| \le 2\epsilon'$. Then we can choose $\epsilon' = \epsilon/2$ to bound the error by $\epsilon$. 
Utilizing the estimate $c^d = e^{-\frac{1}{2}d\nu}$, we conclude with the claimed complexity for constructing the block-encoding. 

Let us write the initial state in the eigenbasis of $H$, i.e.
\begin{equation}
    \ket{\phi} = \sum_{k} \alpha_k \ket{\psi_k},
\end{equation}
where $|\alpha_0| = \gamma >0$. Then applying this block-encoding to the state $\ket{+}\ket{\phi}$ yields 
\begin{align}
    c^d\left(\ket{+}(F(\cos H)+E)\ket{\phi}\right) &= c^d(\ket{+}(F(\cos H)+E) \sum_{k} \alpha_k \ket{\psi_k}) \\
    & =c^d(\ket{+}\sum_k \alpha_k F(\cos \lambda_k)\ket{\psi_k}+E\ket{\psi_k})\\
    & = c^d (\ket{+}(F(\cos H)+E)\ket{\phi}+\ket{-}\ket{\perp}),
\end{align}
where $\ket{\perp}$ is some unnormalized quantum state. Measuring the ancilla register in $\ket{+}$ gives the state $(F(\cos H)+E)\ket{\phi}/\|(F(\cos H)+E)\ket{\phi}\|$. This state is an $2\epsilon$ approximation of $F(\cos H)\ket{\phi}/\|F(\cos H)\ket{\phi}\| \approx \ket{\psi_0}$ with the success probability per run as
\begin{equation}
   p = c^{2d}\|(F(\cos H))\ket{\phi}\|^2\gamma^2(1-\delta) \ge \frac{c^{2d}\gamma^2}{2}.
\end{equation}
To bound the success probability below by $1-\delta$, it suffices to run the algorithm
\begin{equation}
    K = \mathcal{O}\left(\frac{1}{c^{2d}\gamma^2}\log(\frac{1}{\delta})\right)
\end{equation}
times. Plugging in the estimate of $c^d$, we get the scaling for preparing the ground state. Further, we can use amplitude amplification to reduce the complexity. Through amplitude amplification, the repetition time $K = O(\frac{1}{p})$ can be quadratically reduced to $O(\frac{1}{\sqrt{p}})$, then to reach $1-\delta$ success probability, we should repeat 
\begin{equation}
    K = \mathcal{O}\left(\frac{1}{c^{d}\gamma}\log(\frac{1}{\delta})\right)
\end{equation}
runs. Plugging in the estimate of $c^d$ leads to the claimed complexity. 

For estimating the observables, we note that
\begin{equation}
    F(\cos H) = g(\cos H) + E',
\end{equation}
where $\|E\| \le \epsilon'$. Then we have
\begin{align}
    \bra{\psi_0}O\ket{\psi_0} & = \frac{1}{\gamma^2}\bra{\phi}F(\cos H) O F(\cos H)\ket{\phi}\\
    & = \frac{1}{\gamma^2} \left(\bra{\phi}g(\cos H) O g(\cos H)\ket{\phi}+E\right),
\end{align}
where $E = 2\bra{\phi}g(\cos H)OE'\ket{\phi}+2\bra{\phi}E'OE'\ket{\phi} \le 2(\epsilon'+\epsilon^{'2})$. Therefore, to bound the observable estimation error by $\epsilon $, it suffices to choose $\epsilon' = \epsilon\gamma^2/8$ and estimate $\bra{\phi}g(\cos H) O g(\cos H)\ket{\phi}$ with error at most $\epsilon\gamma^2/2$. According to Theorem~\ref{thm:noisy_QSP_observable},
\begin{equation}
     \mathcal{O}\left( \frac{1}{(\epsilon\gamma^2)^2 \left( \prod_{j=1}^d \mathbb{E} \cos e_j \right)^4 } \log\left(\frac{1}{\delta}\right)\right)
\end{equation}
independent runs are needed, leading to our claimed complexity.

\end{proof}

\end{document}